\documentclass{article} % For LaTeX2e
\usepackage{iclr2025_conference,times}

\usepackage{hyperref}
\usepackage{url}

\usepackage{graphicx}
\usepackage{amsmath}
\usepackage{amssymb}
\usepackage{mathtools}
\usepackage{amsthm}

\newcommand{\mbE}{\mathbb{E}}
\newcommand{\mbS}{\mathbb{S}}

\newcommand{\mbH}{\mathbb{H}}

\newcommand{\mbL}{\mathbb{L}}
\newcommand{\mbO}{\mathbb{O}}

\newcommand{\mcB}{\mathcal{B}}

\newcommand{\mcD}{\mathcal{D}}

\newcommand{\mcF}{\mathcal{F}}
\newcommand{\mcH}{\mathcal{H}}

\newcommand{\mcN}{\mathcal{N}}

\newcommand{\mbR}{\mathbb{R}}

\DeclareMathOperator{\argmin}{argmin}

\theoremstyle{plain}
\newtheorem{theorem}{Theorem}[section]

\newtheorem{corollary}[theorem]{Corollary}
\theoremstyle{definition}
\newtheorem{definition}[theorem]{Definition}

\theoremstyle{remark}

\graphicspath{{images}}

\title{Gaussian Differentially Private Human Faces 
          Under a Face Radial Curve Representation}

\author{Carlos Soto\\
Department of Mathematics and Statistics\\
University of Massachusetts Amherst\\
Amherst, MA 01003, USA \\
\texttt{carlossoto@umass.edu} \\
\And
Matthew Reimherr   
\\
Department of Statistics\\
Pennsylvania State University\\
University Park, PA 16802, USA \\
\texttt{mlr36@psu.edu} \\
\And
Aleksandra Slavkovic %\thanks{ Use footnote for providing further information
\\
Department of Statistics\\
Pennsylvania State University\\
University Park, PA 16802, USA \\
\texttt{abs12@psu.edu} \\
\And
Mark Shriver \\
Department of Anthropology\\
Pennsylvania State University\\
University Park, PA 16802, USA \\
\texttt{mds17@psu.edu} \\
}

\iclrfinalcopy % Uncomment for camera-ready version, but NOT for submission.
\begin{document}

\maketitle

\begin{abstract}
 In this paper we consider the problem of releasing a Gaussian Differentially Private (GDP) 3D human face. The human face is a complex structure with many features and inherently tied to one's identity.  Protecting this data, in a formally private way, is important yet challenging given the dimensionality of the problem. We extend approximate DP techniques for functional data to the GDP framework. We further propose a novel representation, face radial curves, of a 3D face as a set of functions and then utilize our proposed GDP functional data mechanism. To preserve the shape of the face while injecting noise we rely on tools from shape analysis for our novel representation of the face. We show that our method preserves the shape of the average face and injects less noise than traditional methods for the same privacy budget. Our mechanism consists of two primary components, the first is generally applicable to function value summaries (as are commonly found in nonparametric statistics or functional data analysis) while the second is general to disk-like surfaces and hence more applicable than just to human faces.

\end{abstract}

\section{Introduction}
\label{s:intro}
In statistical analyses, data and parameters appear in varying degrees of complexity, from simpler forms such as scalars and vectors to more complex such as spherical or hyperbolic, for instance. The structural constraints inherent to data need to be respected throughout any analysis as has been shown in the ``intrinsic statistics" frameworks \citep{pennec2006intrinsic, bhattacharya2003large} for accurate estimation and to preserve said structure.
%Data are observed in a variety of forms such as scalars, tensors, spherical, or hyperbolic, for instance, the latter having more structural constraints. Whatever structure is inherent to our data, the subsequent analysis should be cognisant of and respect such structure as has been shown in the ``intrinsic statistics" frameworks \cite{pennec2006intrinsic,bhattacharya2003large}.
Complex data structures tend to come hand in hand with complex statistical computations and hence the techniques for handling such data have not been widely studied outside of specific scenarios. 
Further, the sheer amount of data that is captured from individuals has increased significantly, and of course, this produces a growing concern for one's ``privacy". In this paper, to support broader sharing of confidential data, we propose releasing a \textit{Gaussian Differentially Private}, GDP, average 3D human face.

\textbf{Motivation and Related Literature:}
Data that live in nonlinear spaces can be challenging to work within the DP framework, as shown in the context of manifolds in \citet{NEURIPS2021_6600e06f,soto2022shape,utpala2022differentially}, private Riemannian optimization in \citet{han2022differentially}, and Gaussian DP on manifolds \cite{jiang2023gaussian}. Preservation of structure has also been considered in the context of private covariance estimation for linear regression in \citet{sheffet2019old} and private principal component analysis in \citet{chaudhuri2013near} which is connected to the Stiefel manifold. For privacy, our proposed ``face radial curve" representation of a human face are functions extracted from a disk and hence lends itself to be examined under the lens of private functional data analysis (FDA) which has been considered in\citet{wasserman2010statistical} and \citet{mirshani2017establishing}, and references therein. 

With our faces constantly being captured (e.g., at a grocery store self-checkout), one could question the privacy protections in place of the collected data. Further, there is a vast amount of literature on identification of individuals from facial data in the area of ``face recognition," see, for instance, the expansive literature review in \citet{kortli2020face}. Also, perhaps quite surprisingly, \citet{venkatesaramani2021re} and \citet{klimentidis2009estimating}  showed that one could identify genomic information from a person's 2D face image; the former study further showed that adding noise to said 2D images can help protect this re-identification. For 2D images, two common practices to attain privacy of faces are blurring and pixelation \citep{li2021deepblur,vishwamitra2017blur}. A major goal in these previous works is privacy in form of \textit{anonymity} which is usually measured by re identification of individuals; our work, however, defines privacy by satisfying the conditions of \textit{differential privacy}. This distinction is necessary as the former, generally, considers successful privacy by lack of re identification while preserving utility while for the latter differential privacy is a property of our mechanism. We do not work with 2D images here, however, the need for privacy is still present. The need for privacy for 3D faces is similar to that in 2D images. Anthropological studies are often interested in identification of DNA using labeled 3D faces \citep{sero2019facial}, the connections between DNA genotype and the associated facial phenotype \citep{white2021insights,naqvi2021shared,weinberg2019hunting}, and average faces for demographics such as age and race classification \citep{tokola20153d}. %Capturing 3D face data is nontrivial and one would ideally want to or may be required to release it, protecting the people in said studies is an important
In such studies, one might want to release the data, or may even be required to, so offering provable confidentiality  guarantees for people in the studies becomes an important task. 

To generate our representation we rely on tools from {\it shape analysis}. The earlier forms of shape analysis, such as Kendall's shape space \citep{kendall1984shape}, are limited to only representing a shape as a finite point cloud. The  field has expanded since to consider more complex data structures such as continuous curves, both planar and space, as in \citet{trouve2005local,klassen2004analysis,srivastava2010shape}, and surfaces as in \citet{jermyn2017elastic,su2020shape}. Human faces have been considered in this space in a similar, yet subtly different, manner such as in \citet{samir2006three,drira2010pose} in which faces are represented as a set of curves ``facial curves" and ``radial curves", respectively. These methods represent a face with curves that are generated independently of each other while we use an entire disk-parameterization to capture features across all faces simultaneously.

\textbf{Main Contributions:}
We develop a novel representation for a collection of 3D faces via a set of curves which we extract from disk parameterizations and refer to these as \textit{face radial curves}. We construct the face radial curves using tools from statistical shape analysis in the interest of preserving the \textit{shape} of an average face during the data sanitization process.
Further, we extend existing approximate DP FDA tools into the Gaussian DP framework \citep{dong2019gaussian}, a recent notion of privacy with attractive properties including ``tight" composition. Under our FDA Gaussian DP mechanism, we generate a private average face of a collection of faces under our representation. While we use our mechanism and representation for faces to address the inherent data privacy concerns, this same methodology can be applied to any surfaces diffeomorphic to a unit disk such as terrain models and additionally any settings where one wants private functional statistical summaries.

%Given the human face is a complex structure which contains both flatter areas such as the cheeks and intricate folds such as the nose, they pose a challenge for generating private estimates as shown previously in \citet{mirshani2017establishing}. We demonstrate a way to leverage the natural structure of the face by the construction of the curve representation. 
% Further, as data can be gathered for a public release, there is a need for doing so in a way that ensures some degree of protection for the subjects. 

%We focus on the task of releasing a private face from the perspective of functional data analysis, which has been attempted previously, but rather than sanitizing the entire face simultaneously as in \cite{mirshani2017establishing} we sanitize disjoint collections of curves which have relatively simple structure drawn from the facial recognition and shape analysis literature.

%-------------------------------------------------------------------------

\section{Notation and Background on Face Representation}
%To achieve privacy in our framework, 
%We require each face to be represented
%as a set of functions which we extract from a parameterized surface. 
We require a disk-parameterized representation of each face from which we extract a set of functions. We parameterize each face independently but ``align" and ``register" each face to a template. Here, we broadly describe these relevant techniques from shape analysis including their necessity.
%We do not present any algorithms for the parameterization process as it is not novel nor uncommon in shape analysis; however, 
To the best of our knowledge, the use of parameterized surfaces in the context of DP has not been explored. The software to accomplish parameterization, registration, and alignment are fully described in \citep{jermyn2017elastic}, with accompanying implementation at GitHub repository \citep{LagaGitHub}. Further, in Figure \ref{fig:diagram} we display the entire pipeline for our methodology.

% then finally we sanitize through a functional data analysis approach. In this section we expand on each of these steps.
%-------------------------------------------------------------------------
\subsection{Parameterization}\label{ss:Param}
Each face is realized as a point cloud, a set of $p$ many  points in $\mbR^3$. %, %$PC\in \mbR^{p\times 3}$, 
We describe the data collection process in \ref{ss:Data}. %In general, $p$ need not be the same for each face. 
%but for the elastic shape analysis framework we require a parameterized surface. 
%For a space curve $f$, a canonical domain for parameterization is a unit interval yielding a map $f:[0,1]\rightarrow \mbR^n$;.
A point cloud does not explicitly relay structural information; i.e., there is no natural ordering nor explicit connectivity or relationship between points. Connectivity can be difficult to infer since if two points are close in space, they may not be neighboring points, e.g., two points near the tip of the nose, or any concave feature, can be close in $\mbR^3$, but based on measuring distance on the surface of the face they may be relatively far apart. So, rather than treat a face as a set of independent points, we can jointly model each face by imposing a parameterization. Specifically, we represent each face as a disk-parameterized surface \citep{jermyn2017elastic,sun20223d,drira2010pose}.%, and to achieve such a parameterization we require a few intermediary steps.

\begin{figure}
    \centering
    \begin{tabular}{@{}c@{} @{}c@{} @{}c@{} }
        \fbox{\includegraphics[trim = 20 50 20 30, clip, height = 1.1in]{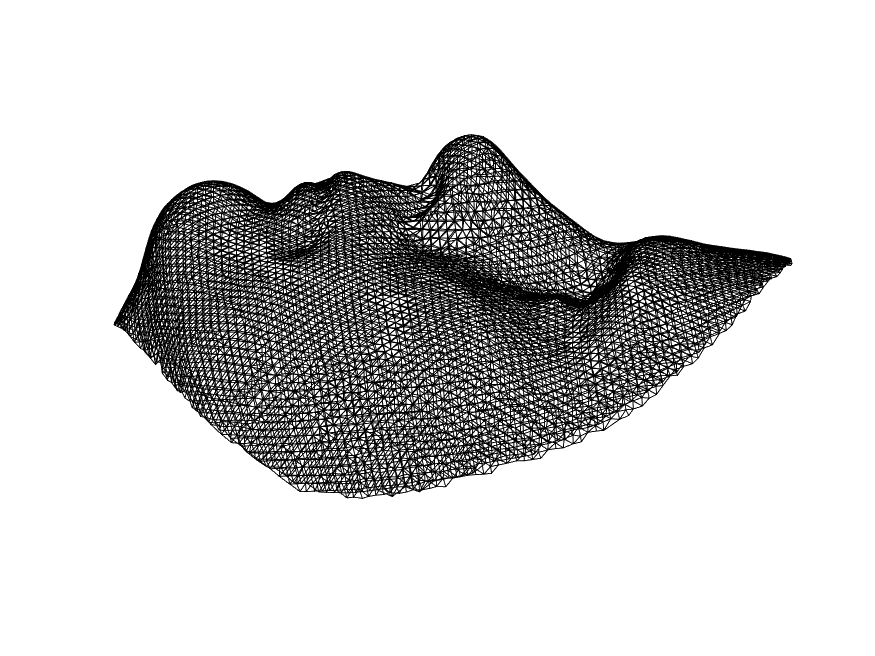}}  & 
        % \fbox{\includegraphics[trim = 85 0 60 0, clip, height = 1.5in]{ConfMap0010.png}} & 
        % \fbox{\includegraphics[trim = 85 0 60 0, clip, height = 1.5in]{ConfMapMobius0010.png}}
        % \\
        \fbox{\includegraphics[trim = 20 50 20 30, clip, height = 1.1in]{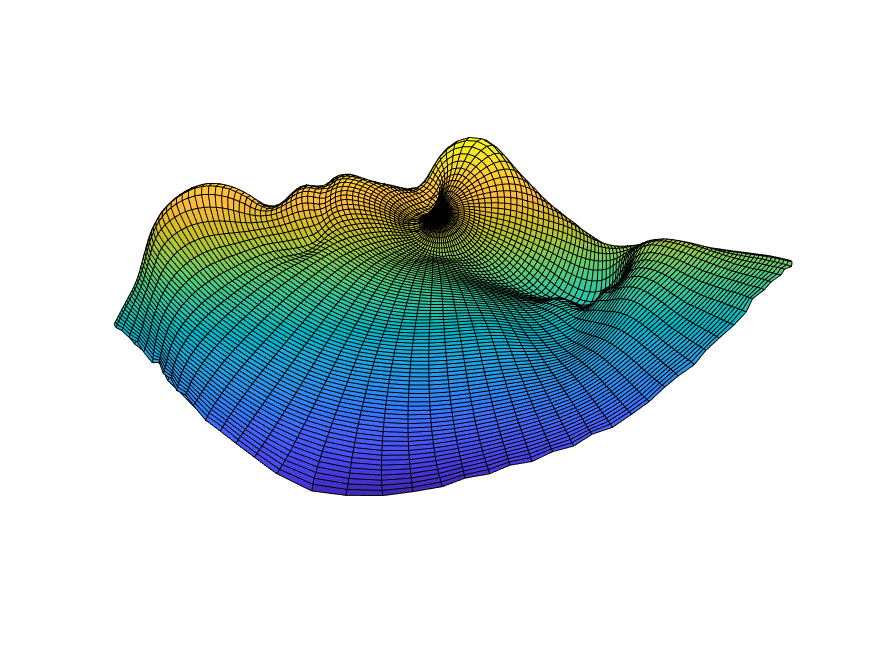}} &
        \fbox{\includegraphics[trim = 20 50 20 30, clip, height = 1.1in]{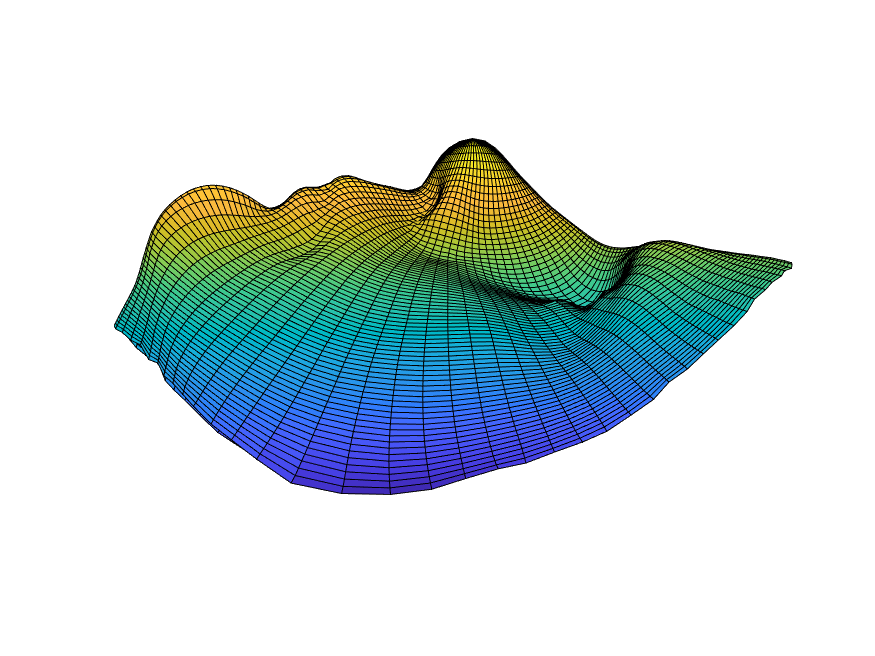}} \\
    \end{tabular}
    \caption{Left: A triangulated mesh face. Middle: A disk-parameterized face, the center of the disk being the left nostril. Right: A disk-parameterized face after a M\"obius transformation forcing the center of the disk to be the tip of the nose.}
    \label{fig:FaceParam}
\end{figure}

We first compute a Delaunay triangulation using the MeshLab software \citep{LocalChapterEvents:ItalChap:ItalianChapConf2008:129-136} yielding a triangulated mesh. Given a point cloud $X\in\mbR^{p\times 3}$ the triangulated mesh, $M=\{V,E,T\}$, is an object which consists of a set of vertices $V$ (the original points), edges $E$, and triangles $T$ composed of said vertices and edges. The triangles meet at edges, do not overlap, and jointly represent a mesh or surface. The left panel of Figure \ref{fig:FaceParam}  displays a face as a triangulated mesh. We generate a \textit{disk conformal map} of $M$ as in \citet{choi2018linear}; a disk conformal map of the triangulated mesh is a mapping from said mesh to a disk which preserves the angles of the triangles or, more generally, the local geometry of the mesh. %The left panel of Figure \ref{fig:ConfMap} displays a disk conformal mapping of the face from the left panel of Figure \ref{fig:FaceParam}.%; it is difficult to discern features, but one can see approximately five dense areas of which the middle is the nose.

%Lastly, we generate a disk-parameterized surface as in \citet{jermyn2017elastic,laga20183d} which requires a disk conformal mapping. Figure \ref{fig:FaceParam} displays two parameterized surfaces in the middle and bottom panels. The middle panel and bottom panel are generated using the disk conformal maps before and after the M\"obius transformation, respectively. While the center of the disk is arbitrarily and stochastically chosen, it is natural to designate it with a feature such as the tip of the nose. We also note that all three panels of the figure have the same shape but have different representations; that is to say, shape is not effected by its parameterization nor representation.

Lastly, we generate a disk-parameterized surface as in \citet{jermyn2017elastic,laga20183d} via the disk conformal mapping. Figure \ref{fig:FaceParam} displays two disk-parameterized surfaces in the middle and right panels; we expand on the differences between these two panels in \ref{ss:Mobius}.
Each face, $f$, is thus a map $f:D\rightarrow \mbR^3$ with $D=\{r,\theta| 0\leq r\leq 1, 0\leq \theta< 2\pi \}$, the unit disk.

%We also note that we did not present any algorithms this section as no part is novel nor the focus on this work, we refer to the references for the interested reader.

%One does not require a disk-parameterization for radial curves, however we wish to jointly register all curves and we expand on this in \ref{ss:esa}.

%In this context, each $r\in [0,1]$ represents a face radial curve.

%While our framework can search over the reparameterization 

%To achieve such a paramterization we must process the data. An arbitrarily sized interval or disk as opposed to unit sized interval or disk would add unnecessary constants and thus we forgo this an option and use the unit disk.
%We parameterize each point cloud to be a surface with a unit disk being the domain of the parameterization. That is, each face $f:D^2 \rightarrow \mbR^3$

%-------------------------------------------------------------------------
\subsection{Shape Analysis of Surfaces}\label{ss:esa}
Figure \ref{fig:FaceParam} displays different representations of a face while retaining its shape. In our setting, the shape of an object is that which is not affected by its scale, location, rotation, or parameterization. The analysis we intend to do should not be dependent on any of these parameters. Next, we describe how to remove this dependence.

Let $\mcF$ be the space of all parameterized surfaces, $\mcF\ni f:D\rightarrow \mbR^3$ where $D$ is the unit disk. We assume all surfaces are smooth and genus-0, that is, they are differentiable and have no holes. We first remove scale and location differences by forcing each face to have unit surface area and be centered at the origin. That is, we scale the surface area to be one by setting $f\rightarrow f/\int_D |f_r\times f_{\theta}|_2 drd\theta$ where $f_r, f_{\theta}$ are the partial derivatives of $f$ with respect to $r$ and $\theta$, respectively. % and thus $|f_r\times f_{\theta}|_2$ is the instantaneous surface area of the surface at ($r,\theta$). 
To center the surface, set $f\rightarrow f-\bar{f}$ where $\bar{f}$ is the centroid of the surface. For notational simplicity, %, and with some abuse of notation, 
let $f$ denote a surface which has unit surface area and the origin as its centroid.

Location and scale are characteristics that are intrinsic to each surface, so we achieve their removal on each face independently. Rotational alignment and parametric registration, however, are relative. Each face must be aligned and registered to a template face.\footnote{In shape analysis it is more typical to do pairwise alignment and registration \citep{wallace2014pairwise,cho2019aggregated,klassen2004analysis} however our goal is not to do pairwise comparisons so we forego this approach.} Let $f_{temp}$ be a template face of unit area and origin centroid; in practice one could either choose an arbitrary surface from the dataset, a training set, or some representative surface such as the mean surface.

Let $\mbO=\{O|\text{det}(O)=1\}$ be the set of all $3\times 3$ rotation matrices and let $\Gamma=\{\gamma|\gamma:D\rightarrow D\}$ be the set of diffeomorphisms. Each $\gamma$ is a reparameterization of a surface which acts on a surface on the right as $f\circ\gamma$; we fully describe this action in \ref{ss:reg}. Rotations, $\mbO$, act on the left as $O\cdot f$ and do not effect scale nor centroid.%, that is $\|f\|_2=\|O\cdot f\|_2$ for all $O\in\mbO$ with $\|f\|^2_2=\int_D |f|_2^2 drd\theta$.

We optimally align each face to the template, $f_{temp}$,  by considering the optimization $$(\tilde{\gamma},\tilde{O}) = \argmin _{\gamma\in\Gamma, O\in\mbO}\|Of\circ\gamma-f_{temp}\|_2.$$ To accomplish this task we appeal to \textit{elastic} shape analysis (ESA) \citep{jermyn2017elastic}. In short, rather than solve the above minimization, ESA defines a metric which replaces the above $\mbL^2$ norm and transforms each surface using the \textit{square-root normal field} representation. For each face we compute this minimization such that the registered and aligned face is $\tilde{f}:=\tilde{O}f\circ\tilde{\gamma}$. We elaborate on this process in \ref{ss:reg}.

%The machinery needed accomplish this task is out of the scope of this paper but we nevertheless summarize it in the Appendix. 
%code available at the authors GitHub and the surface parameterization \cite{jermyn2017elastic} code available on their GitHub as well. 
%https://github.com/garyptchoi
%https://github.com/hamidlaga
% There are other ways one could accomplish this alignment, for instance, using a Procrustes analysis to find an optimal rotation. 

%-------------------------------------------------------------------------
\section{Novel Face Radial Curves Representation}\label{ss:FRCP}
%Rather than work with the entire face at once, we simplify the problem by subsampling a disk-parameterized surface into curves akin to aludw
We parameterize the entire face using a disk, however, our goal is to work with what we refer to as ``face radial curves." This concept is  similar to that of \textit{facial curves} \citep{samir2006three} and \textit{radial curves} \citep{drira2010pose}. In the facial recognition literature, for instance, to construct a radial curve representation one picks a focal point, typically the tip of the nose, and overlays curves on the surface of the face with the nose as the center and each point of the curve being equidistant from the tip of the nose. That is, given a focal point $p$ each radial curve is $g=\{x|r= d_F(p,x)\}$ for a given $r\geq 0$ and where $d_F$ is distance measured on the surface of the face.

%An issue with these previous methods is that the distribution of features on faces is not the same for everyone.
These previously mentioned methodologies do not take into account that the distribution of features on faces is not uniform for everyone. That is, let $g_{ij}$ represent the $j$th radial curve of the $i$th face, then $g_{ij}$ may be the curve which goes directly on the middle of the eyes for individual $i$ but $g_{i'j}$ may be directly on the forehead of individual $i'$ for the same $j$. Features, such as eyes and the mouth, being disproportionately farther or closer to the focal point or each other hence causes issues. These methodologies do, however consider registration within each $j$
%Therefore, registration has been considered within each set of curves $j$ 
but not across all $j$s simultaneously. We propose a new method to have an entire global alignment and registration across all sets of curves and faces.
%$$(O,\gamma)=\argmin_{O,\Gamma}\|g_{ij}-Og_{i'j}\circ\gamma\|$$

%(MOVED THE FOLLOWING FROM A DIFFERENT SECTION, NEEDS TO BE REROEKD)
%Each $f$ is a map of concentric circles which encode the relationship between neighboring points as desired, the image of each $f$ is the surface of the face. Ultimately, rather than focusing on the entire disk at once for privacy, we will focus on ``face radial curves."  We consider $f:D\rightarrow \mbR^3$ to be an infinite collection of concentric circles $f_r:D_r\rightarrow\mbR^3$ where $D_r=\{\theta|r, 0\leq \theta< 2\pi \}$ for a fixed $r$. We elaborate on this distinction and the requirement for a disk parameterization in subsequent sections. 

%We take an approach with a similar flavor but inherently different. 
%In the previous section, we set up a way to compute faces $\{\tilde{f}^{(i)}\}_{i=1,\cdots n}$ which are registered and aligned to a template face $f_{temp}$. Each face is a mapping $\tilde{f}^{(i)}:D\rightarrow \mbR^3$ which is a collection of curves $\tilde{f}^{(i)}_{r}:D_r\rightarrow\mbR^3$. Due to the ``global alignment" of the entire face, the face radial curves $\tilde{f}^{(i)}_{r}$ capture the same features across all faces! % That is, since $\tilde{f}^{(i)}_{r}$ and $\tilde{f}^{(j)}_{r}$ are aligned to $f_{temp}$ for all $i,j$, then $\tilde{f}^{(i)}_{r}$ and $\tilde{f}^{(j)}_{r}$

In the previous section, we set up a way to compute faces $\{\tilde{f}^{(i)}\}_{i=1,\cdots n}$ that are registered and aligned to a template face $f_{temp}$. 
Each face is a mapping $\tilde{f}^{(i)}:D\rightarrow \mbR^3$ and we note a disk is an infinite collection of concentric circles, and thus each $\tilde{f}^{(i)}$ is a collection of curves $\tilde{f}^{(i)}_{r}:D_r\rightarrow\mbR^3$, where $D_r\subset D$ is the circle of radius $r$. By construction, the alignment and registration to a template implies these curves $\tilde{f}^{(i)}_{r}$ capture the same features across all faces; these curves are what we refer to \textit{face radial curves}. %We use $r$ as an index for the radius in the disk and note this has no implications

We apply our proposed method to data described in \citet{sero2019facial}; more details are available in \ref{ss:Data}. The left panel of Figure \ref{fig:facesrad} displays a face as a point cloud and the right panel is the same face but with the disk parameterization and some of the face radial curves overlain. %The right panel of Figure \ref{fig:FaceParam} is also a disk parameterized surface, the face radial curves are the d
While the disk parameterization is an infinite collection of these curves, we pick some number of these curves to represent the face. The number of curves is a tuning parameter where more curves implies, to an extent, more definition in the face. We discuss this further in \ref{ss:FaceExample}.

%In our experiments, 100 radial curves was more than enough and as few as 20 can represent a face without sacrificing too much detail. We tackle privacy for functional data next.

% but more privacy budget that will need to be expended. That is to say, 100 radial curves is much more than enough, but one could use any number of curves; however, if one were to use a thousand radial curves, for instance, one would need to either increase the total budget or distort the face more.
\begin{figure}
    \centering
    \begin{tabular}{ccc}
        {\includegraphics[trim = 100 30 80 20, clip, height =1.4 in]{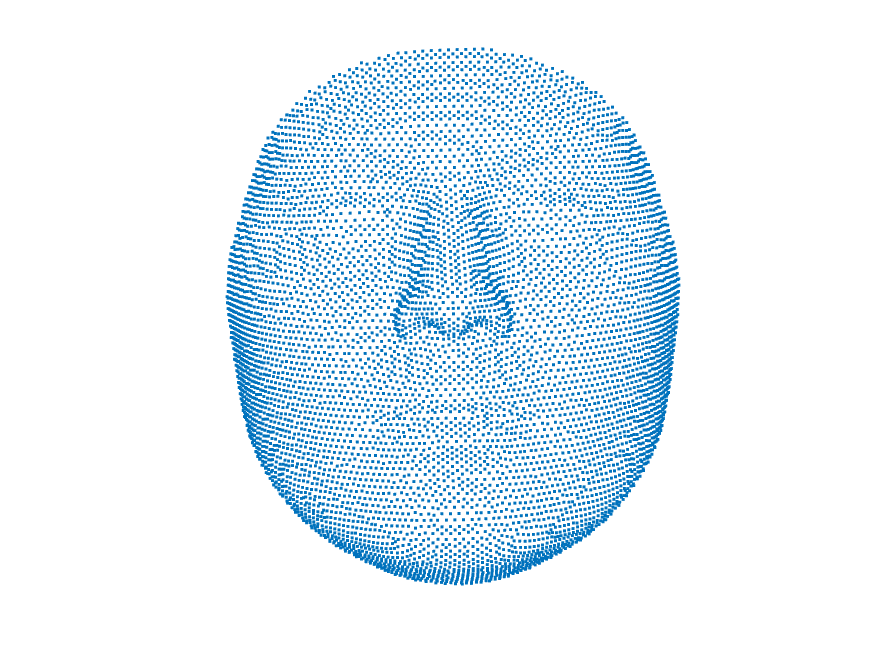} }&
        \includegraphics[trim = 100 30 80 20, clip, height =1.4 in]{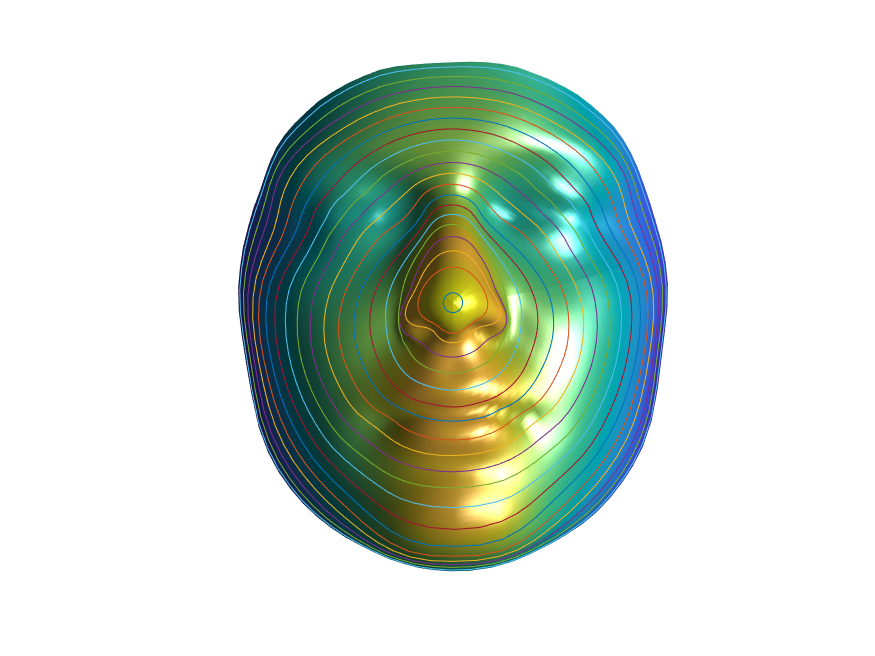}
    \end{tabular}
    \caption{Left: A point cloud representation of a face. Right: A surface representation of the face with an overlay of radial curves.}
    \label{fig:facesrad}
\end{figure}

%-------------------------------------------------------------------------
%-------------------------------------------------------------------------
\section{Gaussian Differential Privacy for Functional Data}\label{ss:GDPFDA}
Since the conception of DP \citep{dwork2006calibrating}, noise calibration has been considered from different perspectives leading to many variants.  
%The original definition of $(\epsilon,\delta)$ can, for instance, be interpreted in the hypothesis testing scenario where $H_0:$ and $H_1:D'$ 
%Besides $(\epsilon,\delta)$-DP,
For instance, ``zero concentrated differential privacy" \citep{bun2016concentrated} and ``R\'enyi differential privacy" \citep{mironov2017renyi} are notions of DP from the perspective of a divergence of the distribution of the mechanism. In the present paper we utilize ``Gaussian DP"  (GDP, \cite{dong2019gaussian}), which considers DP from the perspective of a particular set of hypothesis tests. % $H_0:D$ and $H_1:D'$.
As one of our key contributions in this paper, we extend GDP to functional data analysis (FDA) in Section~\ref{ss:GDPFDA1}. We begin with a brief overview of DP for FDA. %functional data.

\subsection{Background on Differential Privacy for Functional Data}\label{ss:DP}
In this section and in \ref{ss:DPextra}, we describe approximate DP, $(\epsilon,\delta)$-DP, in the context of FDA. This is a needed background for our proposed extension of GDP into the space of functional data. %, and note that everything we describe here is broadly applicable to any function valued statistical summary. 
%In Section \ref{ss:GDPFDA1} we extend GDP into the space of functional data. %, a contribution on its own. 
For a more thorough exposition on DP, see \citet{blum2005practical,dwork2014algorithmic,dwork2006calibrating}. 

Let $D$ denote a dataset of a size $n$, $D=\{x_1,x_2,\dots,x_n\}$, and $\mcD$ be the space of all such datasets. An \textit{adjacent} dataset, $D'$, is a dataset which differs from $D$ in exactly one element which, without loss of generality, we can choose as the last element, $D'=\{x_1,x_2,\dots,x_n'\}$. We write $D\sim D'$ to denote adjacency. 
Here $h(D)$ denotes the statistical summary we aim to release, and a private random version of the summary as $\tilde{h}(D)$, which we will refer to as a privacy mechanism.

We consider releasing an estimate of a private mean function. To sanitize functions we rely on tools and foundations of functional data analysis in the domain of privacy as in \citet{hall2013differential,alda2017bernstein,mirshani2017establishing}, with the latter considering spaces more extensive than functions. The infinite dimensional nature of function valued summaries presents a critical challenge in establishing formal privacy.  In particular, traditional probability densities become much more complicated as there is no default or baseline measure in infinite dimensions (unlike Lebesgue measure in $\mbR^d$).  To overcome this challenge, there are currently two approaches.  The first, taken by \cite{hall2013differential} is to work in finite dimensions and then take careful limits.  The second approach, introduced in \cite{mirshani2017establishing} and which we follow here, is to utilize carefully constructed infinite dimensional densities so that the probability inequalities can be worked out directly (and avoid having to take limits).

%Let $\mathbb{B}$ denote a real separable Banach space equipped with a Borel $\sigma-$algebra. 
% To achieve $(\epsilon,\delta)$-DP one can add Gaussian noise via Gaussian process as in \cite{mirshani2017establishing}. 
Let $\mbH$ denote a real separable Hilbert space in which we aim to release a summary statistic $h(D) \in \mbH$, e.g., $\mbL^2([0,1])$, $\mbR^n$, or a reproducing kernel Hilbert space. We consider a Gaussian process in $\mbH$ to add noise to $h(D)$.
Let $X$ be a Gaussian process in $\mbH$ parameterized by its mean $\eta\in\mbH$ and covariance operator $C:\mbH^*\rightarrow\mbH$ where $\mbH^*$ is the dual space of $\mbH$. A stochastic process is said to be Gaussian if any linear functional $a\in\mbH$ of $X$ is Gaussian in $\mbR$. We note that technically $a\in\mbH^*$, however since $\mbH\cong\mbH^*$, we avoid this distinction unless necessary. Further, we have that $\mbE[a(X)]=a(\eta)$ and $C(a,b)=\text{Cov}(a(X),b(X))$ with $a,b\in\mbH$. We write that $X\sim \mcN(\eta,C)$ and $Z\sim \mcN(0,C)$. In this setting we can achieve approximate DP via Theorem \ref{thm:fdp}. A critical requirement is that our summary be \textit{compatible} with the noise $Z$ which, roughly stated, means that while the noise lives in $\mbH$, we require our summary to exists in a smaller space $\mcH \subset \mbH$. In our case $\mcH$ is a reproducing kernel Hilbert space (RKHS) so the eigenfunctions arise from $C$ and thus are determined by our choice of kernel $k(s,t)$.%, which is a subset of $\mbH$ and where the noise lives.

Now that we have well defined noise, we consider releasing a specific summary statistic: a private mean function with sample mean $h(D)=\bar{X}=\frac{1}{n}\sum_i x_i$. To have some control on the smoothness of this estimate, we enforce smoothness by relying on an optimization problem with a penalty. That is, we let $h(D)=\argmin_{m\in\mcH} \sum_i\|x_i-m\|^2_\mbH+\phi\|m\|_{\mcH}^2$ with $\phi$ being a smoothness penalty parameter. This approach ensures the required noise compatibility. 
%The smoothness of our estimate is one characteristic we can particularly control via a penalized estimate as
For our summary statistic, the sensitivity, $\sup_{D\sim D'}\|h(D)-h(D')\|^2,$ can be shown to be bounded as $\Delta^2\leq 4\tau^2/(n^2\phi)$  \citep{mirshani2017establishing}, where $\tau$ is a finite bound on the norm in the $\mbH$ space of the elements of all $D\in\mcD$.

% \noindent
% The above is summarized below.
% \begin{enumerate}
%     \item Compute the kernel matrix $[K]_{ij}=k(i/m,j/m)$
%     \item Compute the eigen values and vectors $(\lambda_i,v_i)$
%     \item Compute the mean $\hat{\mu}=\frac{1}{n}\sum_i \sum_j\frac{\lambda_j}{\lambda_j+\phi}\langle f_i,v_j\rangle v_j$
%     \item Draw the noise $Z = \sum z\lambda_i v_i$ with $z\sim N(0,1)$
%     \item Compute the private mean $\tilde{\mu}=\hat{\mu}+\Delta*Z$
% \end{enumerate}

%where $g_{ij}=\langle f_i,v_j\rangle$
%-------------------------------------------------------------------------
\subsection{Extension of GDP to Functional Data}\label{ss:GDPFDA1}
One especially useful interpretation of approximate DP comes from \citet{wasserman2010statistical}, which relates DP to hypothesis testing.  In particular, for a sanitized output, $\tilde h(D)$, one can consider statistical tests for determining if the true underlying data source is $D$ or some adjacent data set $D'$.  In these simple cases, the optimal test is well known (Neyman-Pearson Lemma), so one can talk about the optimal type 2 error (1-power) of a statistical test for given type 1 error rate.  It turns out that DP gives a bound for the type 2 error relative to the type 1 error, which \citet{dong2019gaussian} took a step further to define Gaussian DP.  In particular, \citet{dong2019gaussian} show a mechanism is $\mu$-GDP if its entire type 1/type 2 error tradeoff curve is bounded by that of a curve coming from comparing a $N(0,1)$ to $N(\mu,0)$, meaning it is harder to distinguish $0$ and $\mu$ from $N(0,1)$ than it is $h(D)$ from $h(D')$ from $\tilde h(D)$.  We provide the formal definition below. 
\begin{definition}
    A mechanism $\tilde{h}(D)$ is said to satisfy $\mu-$Gaussian differential privacy ($\mu-$GDP) \citep{dong2019gaussian} if for all adjacent datasets $D\sim D'$, $$T(\tilde{h}(x;D),\tilde{h}(x;D'))\geq G_{\mu},$$
    where $T$ is a trade-off function and $G_{\mu}:=T(N(0,1),N(\mu,1))$. Here a trade-off function $T:[0,1]\rightarrow[0,1]$  for two probability distributions $U_1$ and $U_2$ is $T(U_1,U_2)(\alpha)=\inf_\zeta\{\beta_\zeta:\alpha_\zeta\leq \alpha\}$ with $\zeta$ being a rejection rule, $\alpha$ being the type 1 error, and $\beta_\zeta$ the type 2 error for $\zeta$.
\end{definition}
%Informally $\mu-$GDP says that determining the underlying dataset between adjacent datasets is as difficult as differentiating between two random draws from normal distributions with mean difference $\mu$ and unit variance. 
% \begin{align*}
%     T(M(D),M(D')) &= T(\mu+\sigma Z,\mu'+\sigma Z)\\
%     &= T(\mcN(\mu,\sigma^2 C),\mcN(\mu',\sigma^2 C)) \\
% \end{align*}
In our main Theorem \ref{thm:FDAGDP}, we prove that the DP framework for functions as defined in Theorem \ref{thm:fdp} is $\mu-$GDP; we utilize the following corollary %from \citet{dong2019gaussian} 
to achieve this.
\begin{corollary}[\citet{dong2019gaussian}]\label{cc:GDP_DP}
    A mechanism is $\mu-$GDP if and only if it is $(\epsilon,\delta(\epsilon))-$DP for all $\epsilon\geq 0$, where $\delta(\epsilon)\geq \Phi\left(-\frac{\epsilon}{\mu}+\frac{\mu}{2}\right)-e^{\epsilon} \Phi\left(-\frac{\epsilon}{\mu}-\frac{\mu}{2}\right).$ 
\end{corollary}
Corollary \ref{cc:GDP_DP} also appears in \citet{balle2018improving} in the context of the calibrating the Gaussian mechanism, but in the presented form the corollary is more readily applicable for our setting.
\begin{theorem}\label{thm:FDAGDP}
    %Given $D=\{x_1,x_2,\dots,x_n\}$ where the $x_j$ are functional data $\tilde{h}(D)=h(D)+\sigma Z$ where $\sigma\geq \frac{2\log(2/\delta)}{\epsilon^2}\Delta^2$ is $\mu-$GDP.
    The mechanism $\tilde{h}(D)=h(D)+\sigma Z$ as defined in Theorem \ref{thm:fdp} is $\mu$-GDP with $\sigma\geq \Delta/\mu$. Here $\Delta$ is the same global sensitivity as before.
\end{theorem}
% \begin{proof}
%     %We show that $\tilde{h}(D)$ is $\mu-$GDP by bounding the $\delta(\epsilon)$ in Theorem \ref{thm:fdp} as in Corollary \ref{cc:GDP_DP}. %First, we simply reparameterize Corollary \ref{cc:GDP_DP} with $\sigma=\mu=\Delta/\epsilon$.
%     %We have that $P(\tilde{h}\in A)=$
%     Let $PL=\log\left[\frac{dP(D)}{dQ}/\frac{dP(D')}{dQ}\right]$ be the privacy loss random variable with $P(D)$ and $Q$ as defined in \ref{ss:DP}. Since our mechanism $\tilde{h}(D)$ satisfies $(\epsilon,\delta)-$DP we have that,%$\tilde{h}(D)=h(D)+\sigma Z  \implies $
%     \begin{align}
%            %\delta &\geq P\left(\log \frac{P(\tilde{h}(D)\in A)}{P(\tilde{h}(D')\in A)}\geq\epsilon\right)-e^{\epsilon} P\left(\log\frac{P(\tilde{h}(D)\in A)}{P(\tilde{h}(D')\in A)}\leq\epsilon\right) \\ 
%            \delta &\geq P\left(PL\geq\epsilon\right)-e^{\epsilon} P\left(PL\leq-\epsilon\right) \\ 
%            &= \Phi\left(-\frac{\epsilon\sigma}{\Delta}+\frac{\Delta}{2\sigma}\right)-e^{\epsilon} \Phi\left(-\frac{\sigma\epsilon}{\Delta}-\frac{\Delta}{2\sigma}\right)\\           
%            &= \Phi\left(-\frac{\epsilon}{\mu}+\frac{\mu}{2}\right)-e^{\epsilon} \Phi\left(-\frac{\epsilon}{\mu}-\frac{\mu}{2}\right).
%     \end{align}
% The last equality is a reparametrization setting $\mu=\Delta/\sigma$ and the rest of the details are in the Appendix.
% %The first equality is due to $\tilde{h}(D)=h(D)+\sigma Z$ where $Z\sim N(0,C)$. 
% \end{proof}
\begin{proof}
     We need to show that the distribution induced by our Gaussian process in $\mbH$ of $\tilde{h}(D)$ is $\mu-$GDP by bounding the $\delta(\epsilon)$ in Theorem \ref{thm:fdp} as in Corollary \ref{cc:GDP_DP}. The key here is the use of $\mbH$, as our function space is infinite-dimensional and hence it is not possible to define a measure analogous to that of the Lebesgue measure.
     
     %First, we simply reparameterize Corollary \ref{cc:GDP_DP} with $\sigma=\mu=\Delta/\epsilon$.
    %We have that $P(\tilde{h}\in A)=$
    \citet{mirshani2017establishing} provide the framework to define a useful density of $\tilde h(D)$ over $\mbH$,  however they do not consider bounding the tail probabilities nor the privacy loss random variable as we do next. We first define the privacy loss random variable as $PL=\log\left[\frac{dP(D)}{dQ}/\frac{dP(D')}{dQ}\right]$, where $Q$ is the probability measure induced by our noise $Z$ and $P(D)$ the family of measures of our mechanism $\tilde{h}(D)$. By construction we have that $h(D)$ is compatible with $Z$ and hence the above is well defined:
    $$\frac{dP(D)}{dQ}(y)=\exp\left\{ -\frac{1}{2\sigma}\left[\|h(D)\|_{\mcH}^2-2T_{h(D)}(y)\right] \right\}.$$
    Here $T_{h(D)}$ is a linear operator and thus $T_{h(D)}(\tilde h(D))$ is normally distributed.  
    
    From \citet{mirshani2017establishing}, we have that $T_{h(D)}(\tilde h(D)) - 
    T_{h(D')}(\tilde h(D))
    \sim N(0, \|h(D) - h(D')\|^2_\mcH)
    $.
    %
    %Since our mechanism $\tilde{h}(D)$ satisfies $(\epsilon,\delta)-$DP we have that,%$\tilde{h}(D)=h(D)+\sigma Z  \implies $
    Using this, we can show that
    \begin{align}
           %\delta &\geq P\left(\log \frac{P(\tilde{h}(D)\in A)}{P(\tilde{h}(D')\in A)}\geq\epsilon\right)-e^{\epsilon} P\left(\log\frac{P(\tilde{h}(D)\in A)}{P(\tilde{h}(D')\in A)}\leq\epsilon\right) \\ 
           \delta & :=  P\left(PL\geq\epsilon\right)-e^{\epsilon} P\left(PL\leq-\epsilon\right) \\ 
           &= \Phi\left(-\frac{\epsilon\sigma}{\Delta}+\frac{\Delta}{2\sigma}\right)-e^{\epsilon} \Phi\left(-\frac{\sigma\epsilon}{\Delta}-\frac{\Delta}{2\sigma}\right)\\           
           &= \Phi\left(-\frac{\epsilon}{\mu}+\frac{\mu}{2}\right)-e^{\epsilon} \Phi\left(-\frac{\epsilon}{\mu}-\frac{\mu}{2}\right).
    \end{align}
    %as desired.  
The last equality is a reparametrization setting $\mu=\Delta/\sigma$ and the rest of the details are in the \ref{ss:proof}.
%The first equality is due to $\tilde{h}(D)=h(D)+\sigma Z$ where $Z\sim N(0,C)$. 
\end{proof}
For facial radial curves and other disc-like surfaces, by construction, we need to sanitize many mean curves, so we require composition of a multitude of mechanisms with their respective budgets $\mu_i$'s. Composition in the GDP framework is ``tight." 
Given two mechanisms $\tilde{h}_1$, $\tilde{h}_2$ with privacy parameters $\mu_1$ and $\mu_2$, their composition $(\tilde{h}_1,\tilde{h}_2)$ is $\sqrt{\mu_1^2+\mu_2^2}$-GDP \citep[Corallary 3.3]{dong2019gaussian}.  

%-------------------------------------------------------------------------
\section{Face Radial Curve Example}\label{ss:FaceExample}
 %The $J$ is a tuning parameter where larger $J$ implies more detail of the face.
%In Figure \ref{fig:facedetail} from left to right we show $J$ at 16, 27, and 80. We see that $J$ need not be too large to encapsulate the structures of the face. We treat the $J$ sets of curves across the faces independently. 
Here, we apply our proposed methods from sections \ref{ss:FRCP} and \ref{ss:GDPFDA1} on data described in \cite{sero2019facial} and \ref{ss:Data}. Figure \ref{fig:facedetail} displays from left to right the same face with $J=$ 16, 27, and 80 face radial curves. Larger values of $J$ add definition to the face, but $J$ need not be too large to encapsulate the facial structures. Based on that, we represent the $n=1000$ faces with $J=23$ face radial curves; one could pick $J$ in a data driven way as well. %Our choice of $J$ is ad hoc, but one could pick this in a data driven way. 
We treat the $J$ sets of curves across the faces independently. 

\begin{figure}
    \centering
    \begin{tabular}{@{}c@{}@{}c@{}}
        \fbox{\begin{tabular}{@{}c@{} @{}c@{} @{}c@{}}
            \includegraphics[trim = 80 20 70 10, clip, height=1in]{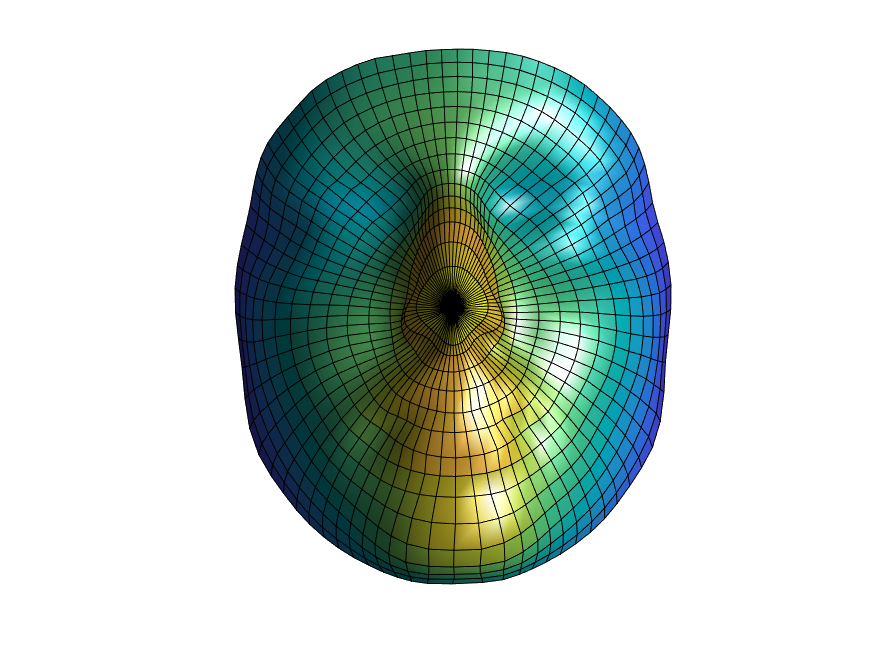} & 
            \includegraphics[trim = 80 20 70 10, clip, height=1in]{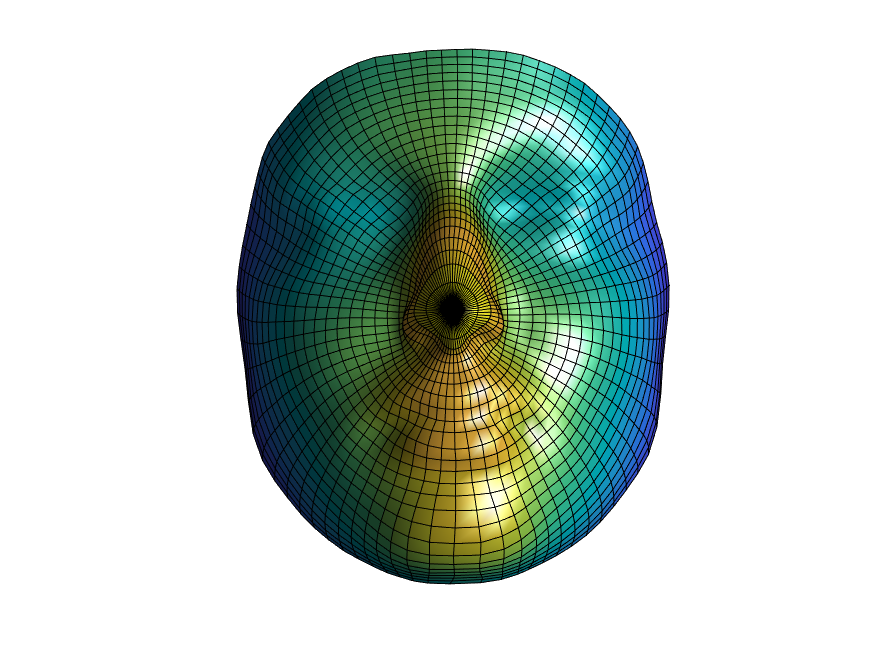} &
            \includegraphics[trim = 80 20 70 10, clip, height=1in]{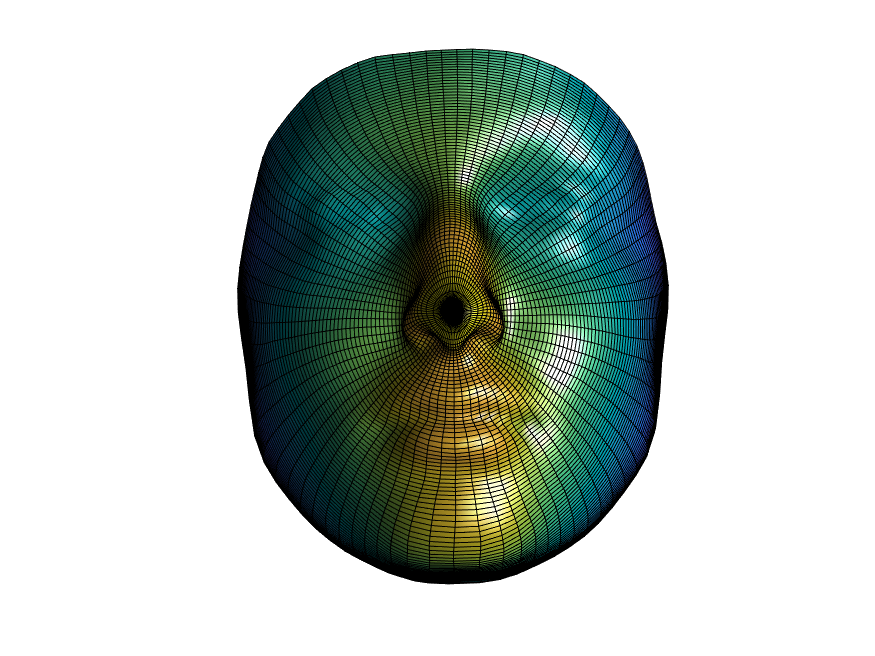} 
        \end{tabular}
        }
         & 
        \begin{tabular}{@{}c@{}}
            \fbox{\includegraphics[trim = 0 0 10 50, clip, width=0.47\textwidth]{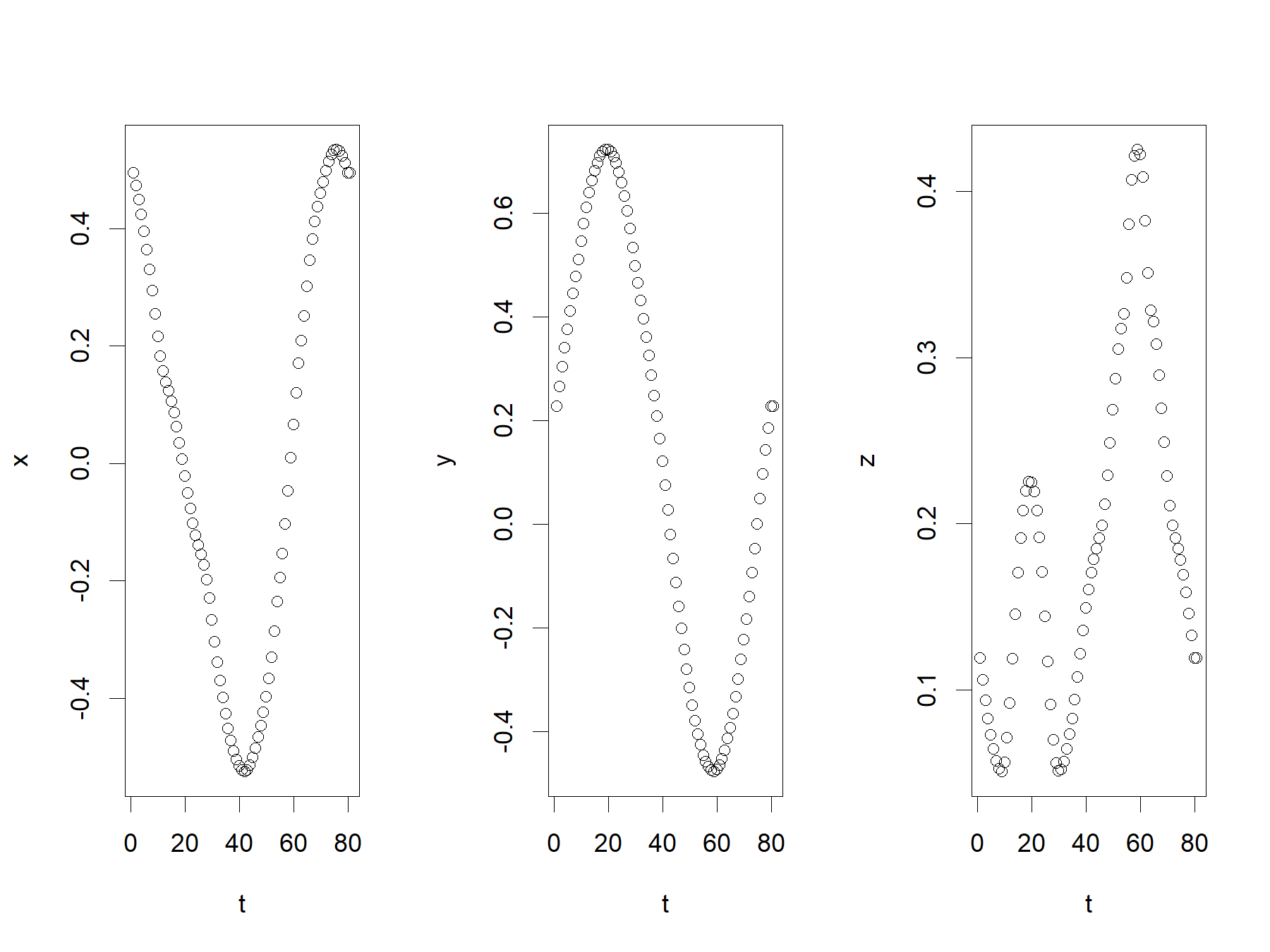}}
        \end{tabular}

    \end{tabular}
    \caption{Left: The same face represented using 16, 27, and 80 face radial curves, respectively. Right: The $x$, $y$, and $z$ coordinate curves, respectively, for a particular face radial curve.}
    \label{fig:facedetail}
\end{figure}

% \begin{figure}
%     \centering
%     \begin{tabular}{@{}c@{} @{}c@{} @{}c@{} }
%         \includegraphics[trim = 80 20 70 10, clip, height=1in]{FaceRadCurve16.png} & 
%         \includegraphics[trim = 80 20 70 10, clip, height=1in]{FaceRadCurve27.png} &
%         \includegraphics[trim = 80 20 70 10, clip, height=1in]{FaceRadCurve80.png} 
%     \end{tabular}
    
%     \caption{Left to right: The same face represented using 16, 27, and 80 face radial curves, respectively.}
%     \label{fig:facedetail}
% \end{figure}

We note an important subtle strength in our construction. Each face radial curve has three coordinates, $x,y,z$, and by our construction $x,y$ are effectively a circle and $z$ encode facial features. The right panel of Figure \ref{fig:facedetail} displays the coordinate curves of an example face radial curve. The first two coordinates are quite simple and are roughly one period of a sine curve. We leverage this simplicity in these two coordinate curves to conserve privacy budget by enforcing a larger smoothing parameter as compared to the last coordinate curve and also treat the coordinate curves separately at each radial curve.

% \begin{figure}
%     \centering
%     \fbox{\includegraphics[trim = 0 0 10 50, clip, width=0.47\textwidth]{figurexyz}}
%     \caption{Left to right: The $x$, $y$, and $z$ coordinate curves for a specific face radial curve.}
%     \label{fig:xyzrad}
% \end{figure}

%to the face however could require more privacy budget to sanitize. 
Let $\{f_i\}$ be the set of radial curves for a specific coordinate at the $j$th position and $i$ being the index of for particular face. For simplicity one can imagine $j=1$ being the curve nearest the tip of the nose and $j=23$ as the curve nearest the border of the face. Each curve $f_i$ is a closed parameterized curve, $f:\mbS^1\rightarrow \mbR$ where $\mbS^1$ is the unit circle, as in the right panels of Figure \ref{fig:facedetail}. %, the one-dimensional sphere. 
In general, we drop the $j$ as we treat each set independent of each other. We parameterize the curve with unit circle but for simplicity we say $f:[0,1]\rightarrow \mbR$ with $f(0)=f(1)$.% and $f(t)=f(t+1)$ for all $t\in(0,1)$. 

By design we have closed curves, so % we should retain this structure throughout the analysis. The choice of kernel is key for this retention, so 
to retain this structure we use a periodic kernel that takes the form %$k(s,t)=\sigma^2 \exp \left(-\frac{2\sin^2(\pi|s-t|/T)}{\rho^2}\right)$ where $T$ is the period and $\rho$ is a parameter on the range of the kernel. 
$k(s,t)=\exp \left(-\left[d_{\mbS^1}(\omega(t),\omega(s)) / \rho \right]^\alpha\right)$ \citep{gneiting2013strictly} with $t,s\in[0,1]$. The distance of two points $a,b$ on $\mbS^1$ is $\arccos\langle a,b \rangle$, however, the previous $t,s$ are parameters on the unit interval, so we first ``wrap" the interval around the circle to compute this distance hence the need for $\omega(\cdot)$, a wrapping function. This kernel is a powered exponential on spheres with kernel range parameter $\rho$ and smoothness parameter $\alpha\in(0,1]$. For our experiments we set $\alpha=1$ as we can control smoothness instead via $\phi$ in the mean estimation as in \ref{ss:DP}.

We compute kernel matrix $K$ of the the unit interval $[0,1]$ with $[K]_{ij}=k(\omega(i/m),\omega(j/m))$ where $m=80$ is an integer defining the fineness of the uniform grid on the unit interval and $i,j=0,1,\dots m$. Let the eigenvalues, of $K$, and their associated eigenfunctions be denoted as $(\lambda_i,b_i)$. It can be shown that non-private mean in the RKHS space is $h(D)=\frac{1}{n}\sum_i \sum_j\frac{\lambda_j}{\lambda_j+\phi}\langle f_i,b_j\rangle b_j$. % where $\phi$ is a smoothness parameter. 
Since we are sanitizing the coordinates independently, we have a smoothness parameter for each coordinate, $\phi_x$, $\phi_y$, and $\phi_z$. The left panel of Figure \ref{fig:meanfaces} displays the mean face constructed from mean face radial curves with $\phi_x=\phi_y=0.01$ and $\phi_z=0.005$.
\begin{figure}
    \centering
    \begin{tabular}{|@{}c@{} @{}c@{}|@{}c@{} @{}c@{}|}
    \hline
        {\includegraphics[trim = 50 20 30 0, clip, height=1.3in]{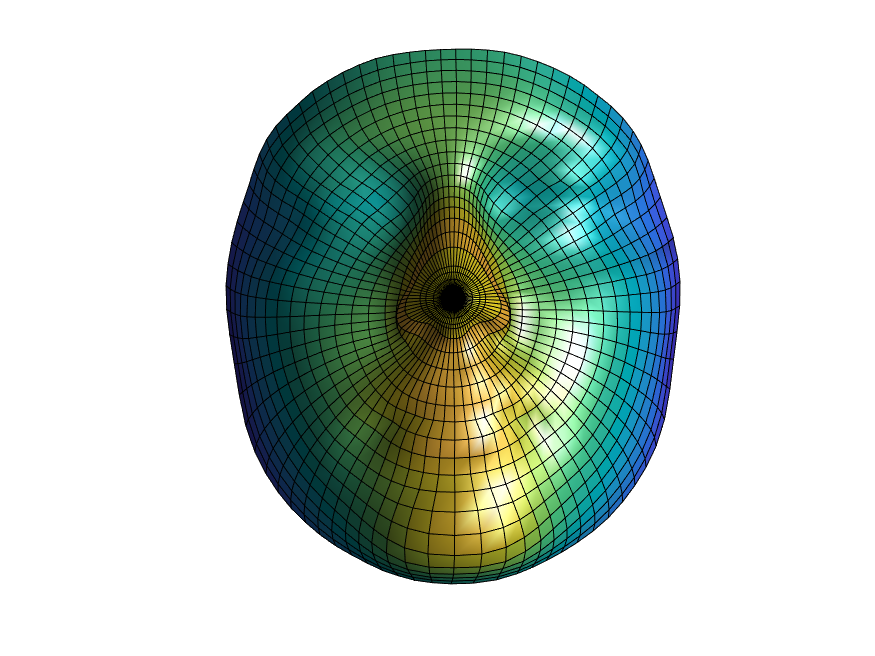}}&  
        {\includegraphics[trim = 100 40 110 20, clip, height=1.3in]{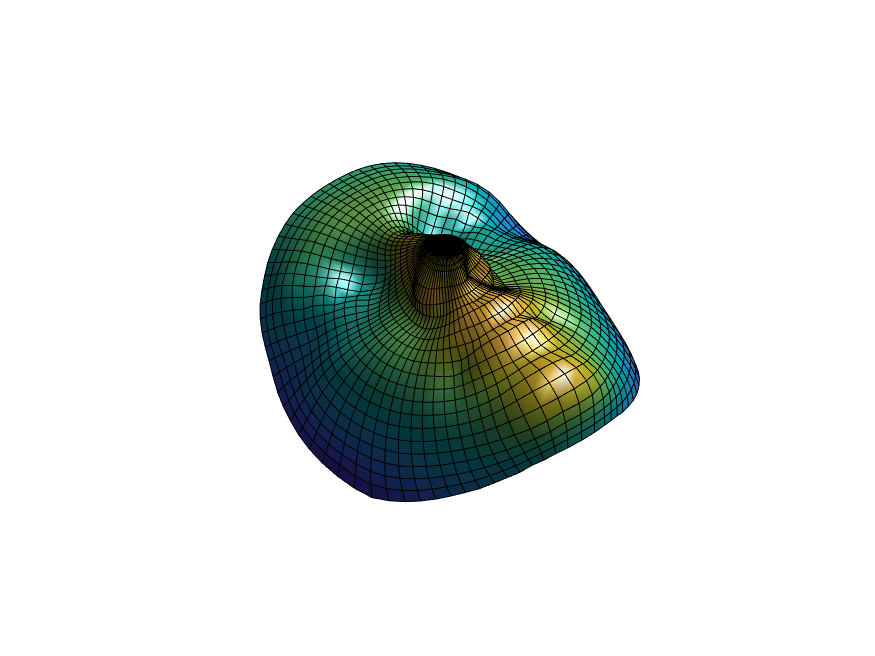} } &%\\\hline
        {\includegraphics[trim = 100 0 60 0, clip, height =1.3in]{pointwisemeanAng1}} &
        {\includegraphics[trim = 100 50 100 50, clip, height =1.3in]{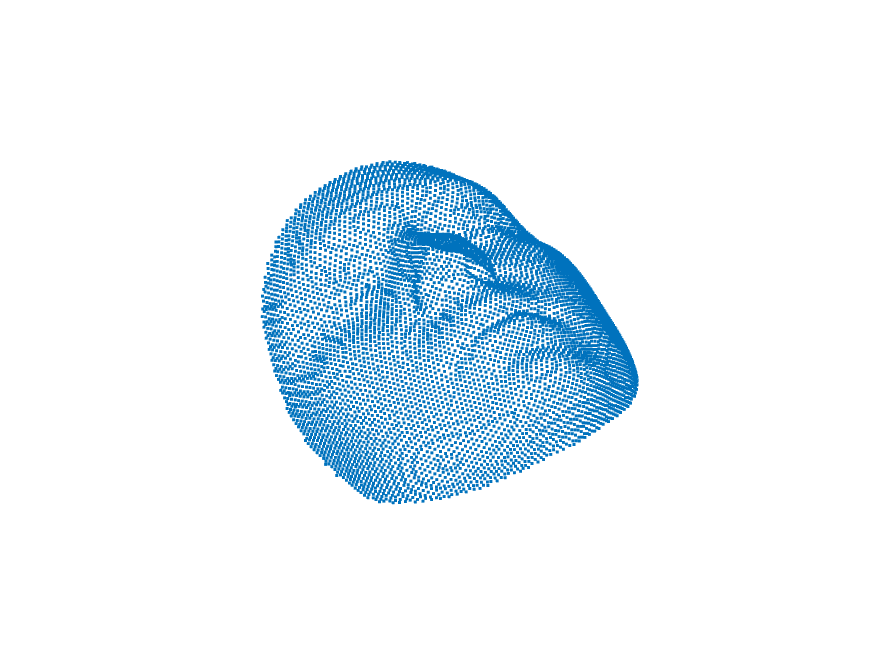}} \\\hline
        
    \end{tabular}
    
    \caption{Left: Two angles of the average face constructed of average face radial curves with $\phi_x=\phi_y=0.01$ and $\phi_z=0.005$. Right: Two angles of the point-wise mean.}
    \label{fig:meanfaces}
\end{figure}

%The private mean $\tilde{\mu}$ is then $\tilde{\mu}=\hat{\mu}+\Delta*Z$ where $\Delta$ is sensitivity, a scalar, and $Z$ is a draw from the normal distribution parameterized with our kernel. That is, we sample $Z$ by $\sum z\lambda_i v_i$ where $z\sim \text{N}(0,1)$ a draw from the standard normal.
%-------------------------------------------------------------------------
\subsection{Experimental Results}\label{ss:ExRes}
To highlight the difficulty of working with such complex data we present some preliminary results. We have 1000 faces each with 7150 points; we fully describe the data in \ref{ss:Data}. The points are ``registered", via the data collection method, across all the faces. Without this registration, the following would not be immediately feasible. Let $D=\{X_i\}$ and $X_i\in\mbR^{7150\times 3}$. %be the $l=1,2,\cdots 7150$ landmark and $w\in(x,y,z)$ coordinate.
We first compute the point-wise mean of the faces, $\bar{X}=1/n\sum_i X_i$ and display this in the right panel of Figure \ref{fig:meanfaces}.

Suppose we have a total budget of $\mu_T$ and we wish to sanitize, independently, the 7150 points each of which has 3 coordinates. To split the budget evenly, each point's coordinate is allocated $\mu_p:=\sqrt{\mu_T^2/(7150\cdot 3)}$ of the budget. %Let $\bar{X}$ denote the point-wise mean, $\bar{X}_{i}$ the $i$th landmrk, and $\bar{X}_{i,w}$ denote the $w$the coordinate of the $i$th landmark. At each location we compute the sensitivity as $\Delta_w=\max_{X}|X_{i,w}-\X_{i,w}|$
For each point and each coordinate we compute sensitivity as $\Delta_{k,l}=\max_{i,j}|X_i[k,l]-X_j[k,l]|$. That is, among all faces this measures the variability at each point. We note that this sensitivity calculation does violate privacy, as it is data driven, and that an entirely private form of this calculation is strictly larger. That is, our sensitivity calculation is smaller, so a private method would lead to noisier estimates. We add noise to each coordinate of each point as $\bar{X}[k,l]+\Delta_{k,l}/\mu_p \cdot z$ where $z\sim \mcN(0,1)$. Figure \ref{fig:pointexamples} displays two angles of a point-wise private face in the left panel with total privacy budget $\mu_T=3$ with more results in \ref{ss:AddRes}. %The left and right panels of each respective row are the same point cloud from two different angles.

% In Figure \ref{fig:pointexamples} on the left column we have the point-wise mean of the faces. To have a point of comparison we see what a point-wise private mean would look like. We use a pure DP approach and add Laplace noise to each coordinate of the mean. In the middle column of the figure we see that with a total $\epsilon=1$ split evenly among the landmarks. We see that this entirely destroys the face as the number of landmarks is large this does not take into account the structure of the face. In the third column of the same figure we again display a point-wise private mean with  $\epsilon=30$, which can be viewed as a large budget. The face starts to retain some structure at this large $\epsilon$ but still distorts areas such as the lips and eyes. 

Next we apply our approach using face radial curves. We have $J$ many sets of curves, $\{f_i\}_j$ each of which has three coordinate curves, %$f_i=(f_i(x),f_i(y),f_i(z))$. 
$f_i=(f_{ix},f_{iy},f_{iz})$.
We sanitize the regularized mean, independently, for each coordinate and each radial curve, using our mechanism which satisfies GDP with sensitivity is $\Delta^2\leq 4\tau^2/(n^2\phi)$. Earlier we saw that the $x$ and $y$ coordinate curves, at every $J$, are effectively one period of a sinusoidal curve, we take advantage of this construction to conserve budget. We thus spend less budget sanitizing the $x$ and $y$ coordinate curves and spend more on $z$ as they contain more feature information. 

\begin{figure}[t]
    \centering
    \begin{tabular}{|@{}c@{} @{}c@{}|@{}c@{} @{}c@{}|}
    \hline
        % {\includegraphics[trim = 100 0 60 0, clip, height =1.3in]{pointwisemeanAng1}} &
        % {\includegraphics[trim = 100 50 100 50, clip, height =1.3in]{pointwisemeanAng2}} \\
        
        % {\includegraphics[trim = 100 0 60 0, clip, height =1.3in]{PWmeanGDPmu2Ang1}} &
        % {\includegraphics[trim = 100 50 100 50, clip, height =1.3in]{PWmeanGDPmu2Ang2}} &
        
        {\includegraphics[trim = 100 0 60 0, clip, height =1.3in]{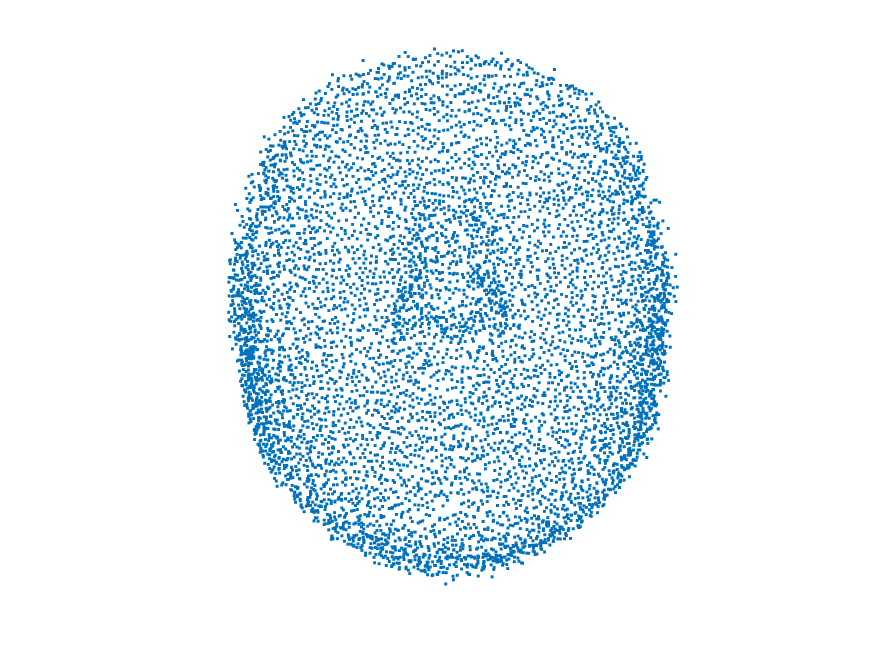}} &
        {\includegraphics[trim = 100 50 100 50, clip, height =1.3in]{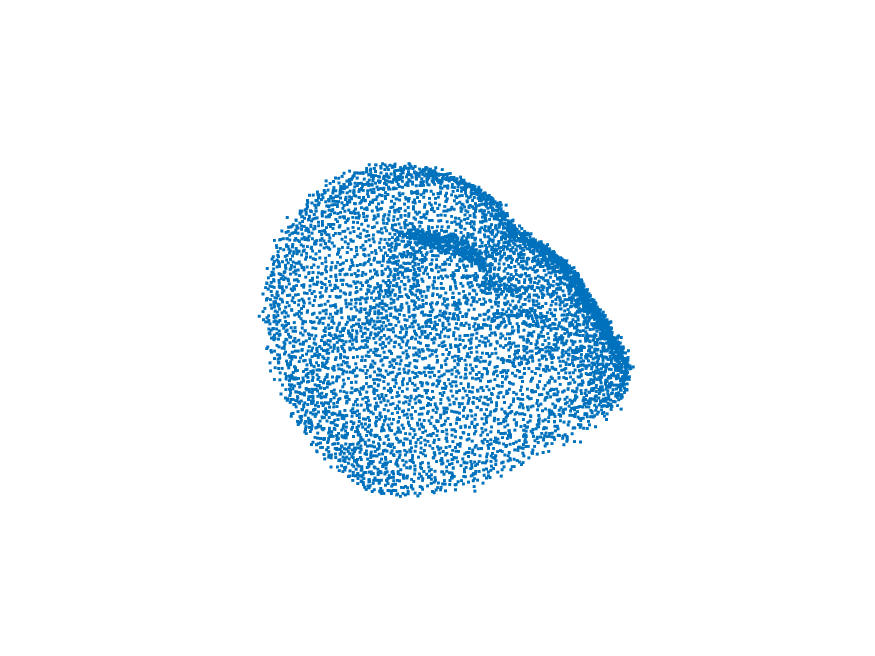}} &  
        
        \includegraphics[trim = 50 20 30 0, clip, height=1.3in]{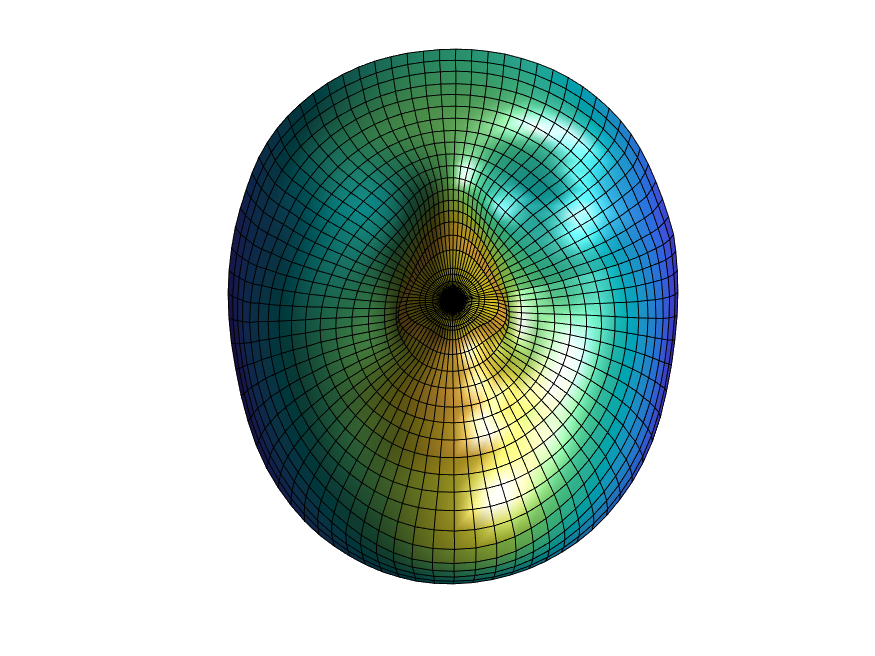} &    
        
        \includegraphics[trim = 100 40 110 20, clip, height=1.3in]{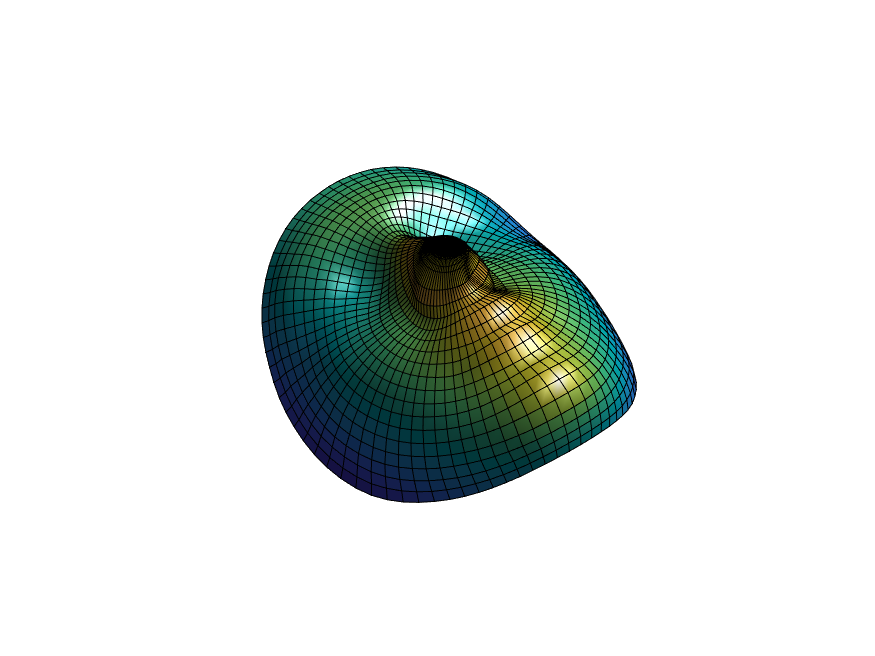} \\\hline
    \end{tabular}
    
    \caption{Left: Two angles of a private face sanitized with point wise Gaussian noise with total $\mu_T=2$. %Second Row: A private face sanitized with Gaussian noise point-wise with total $\mu_T=3$. 
    Right: Two angles of a differentially private face constructed from private mean face radial curves with $\mu_T=2.9961$ $\phi_x=\phi_y=0.01$ and $\phi_z=0.005$.}
    \label{fig:pointexamples}
\end{figure}

We need to determine the $\tau$, an upper bound on the norm of the curves, in our $\Delta$. We determine this $\tau$, in a similar way as the point-wise approach of computing sensitivity, through the data at each set of radial curves and each coordinate. That is, for each $j$ and at each coordinate $w$, $\tau_w=\max_i\|f_i(w)\|_\mbH$. Again, we have a sensitivity that is partially data driven. However we note since we use shape analysis to do our processing, all the faces are scaled to have unit surface area, and hence all the curves are scaled. Thus, individually none of these norms hold any meaningful size information. %To further emphasize this, 
Further, our analysis is entirely \textit{scale invariant}, meaning that if a face is much larger or smaller than average, that information is completely removed when we scale. This is a strength of shape analysis as the emphasis is the \textit{shape} of the face and not the nuisance parameters. 

We choose a budget for each of the coordinates $\mu_x,\mu_y,\mu_z$ which are the same at all levels $j$ of curves; i.e., at each $j$ for the $w$ coordinate curve $\tilde{h}(\{f_{iw}\})=h(\{f_{iw}\})+\frac{\Delta_{w}}{\mu_w}Z$ with $\Delta_w^2\leq 4\tau_w^2/(n^2\phi)$ and $Z$ is a standard normal Gaussian process. The right panel of Figure \ref{fig:pointexamples} displays an example private mean face with $\mu_T=2.9661$ where $\mu_T=\sqrt{J(\mu_x^2+\mu_y^2+\mu_z^2)}$ with $J=23$, $\mu_x=\mu_y=0.2$ and $\mu_z=0.55$. For each coordinate we have separate smoothing parameters we display results with $\phi_x=\phi_y=0.01$ and $\phi_z=0.005$. Comparing our private mean to left panel of Figure \ref{fig:meanfaces} we see there is some clear smoothing in the lips area and some subtle smoothing in the eyes as well. %Further, comparing our estimate to the top two rows of Figure \ref{fig:pointexamples} we see some key differences in smoothness and structure. 

Lastly, we quantify the amount of injected noise by computing a mean squared error. We let the point-wise mean be our baseline non-private estimate, the right panel of Figure \ref{fig:meanfaces}. 
%For each of the two point-wise private estimates, the middle rows of the same figure, the $MSE=\frac{1}{N}\sum_i |\bar{X}_i-\widetilde{X}_i|^2$ where $N=7150$ is the number points. Table \ref{tab:Error} displays these errors in the first two columns. 
For the point-wise private estimate $\tilde{X}$ %, the middle rows of the same figure, 
the $MSE=\frac{1}{N}\sum_i |\bar{X}_i-\widetilde{X}_i|^2$ where $N=7150$ is the number points. Table \ref{tab:Error} displays these errors in the first two columns. Our sanitized mean, however, is a set of functions so we first discretize it into 1863 points (23 curves at 81 points) and denote this as $\widetilde{V}$. We let $MSE = \frac{1}{M}\sum_j \min_i|\bar{X}_i-\widetilde{V}_j|^2$ with $M=1863$. Since the two point clouds have different number of points, this MSE finds the nearest point in the non-private face to that of the private face $\widetilde{V}$. We further scale $\widetilde{V}\rightarrow a\widetilde{V}$ to align it to $\bar{X}$ as %$\widetilde{V}$ is not on the same scale as 
$\bar{X}$ is not processed data. Table \ref{tab:Error} displays the error in the third column of 5.2989, less than both of the point-wise private faces. This MSE does incur an inflation, though, as the point-wise mean is not a surface and thus lacks a registered point at the location of our private mean. We can clearly see this in Figure \ref{fig:ErrorFig}, in both panels the blue points are the point-wise mean, the left panel has the point-wise private mean in red, and the right panel has our private (discretized) mean in red. We see that our method has points that seem to lie on the ``surface" of the face but the non private mean may not have a point there. We also clearly see that the point-wise private mean adds noise in the ambient space and thus creates a rough, fuzzy estimate not entirely resembling a smooth face; i.e., the facial structures are distorted. Also, the MSE of our facial radial curve based GDP method is less than the MSE of the point-wise sanitized estimate which has a slightly larger privacy budget.
To even further emphasize this point, the last column of Table \ref{tab:Error} is the MSE between our private mean and the RKHS mean, we see it is less than half in comparison to the point-wise sanitization.

\begin{table}[]
    \centering
    \caption{Mean Squared Error between private estimates and the point-wise mean. The last column is the point-wise MSE between our method and the RKHS mean. All values are at E-04.}
    \begin{tabular}{|c|c|c|c|c|}
    \hline
    & \multicolumn{4}{|c|}{$\mu_T$}\\ \hline
          &  $2$ &  $3$ &  $2.9961$ (Ours)& Ours*\\\hline
        %Norm Difference     &  2.2305       & 1.5006        & 0.6430        \\\hline
        MSE       &  14       &   7.5807     &  5.2989       & 3.6365\\\hline
        %MAE       &  0.0352       & 0.0239        & 0.0172        \\\hline
    \end{tabular}
    
    \label{tab:Error}
\end{table}
% \begin{table}[]
%     \centering
%     \caption{Mean Squared Error between private estimates and the point-wise non-private mean. The last column is the point-wise MSE between our method and the RKHS mean. All values are at E-04.}
%     \begin{tabular}{|c|c|c|c|}
%     \hline
%     & \multicolumn{3}{|c|}{$\mu_T$}\\ \hline
%           &    $3$ &  $2.9961$ (Ours)& Ours*\\\hline
%         %Norm Difference     &  2.2305       & 1.5006        & 0.6430        \\\hline
%         MSE       &     7.5807     &  5.2989       & 3.6365\\\hline
%         %MAE       &  0.0352       & 0.0239        & 0.0172        \\\hline
%     \end{tabular}
    
%     \label{tab:Error}
% \end{table}

\begin{figure}
    \centering
    \begin{tabular}{@{}c@{} @{}c@{} }
        {\includegraphics[trim = 40 75 20 50, clip, height = 1.2in]{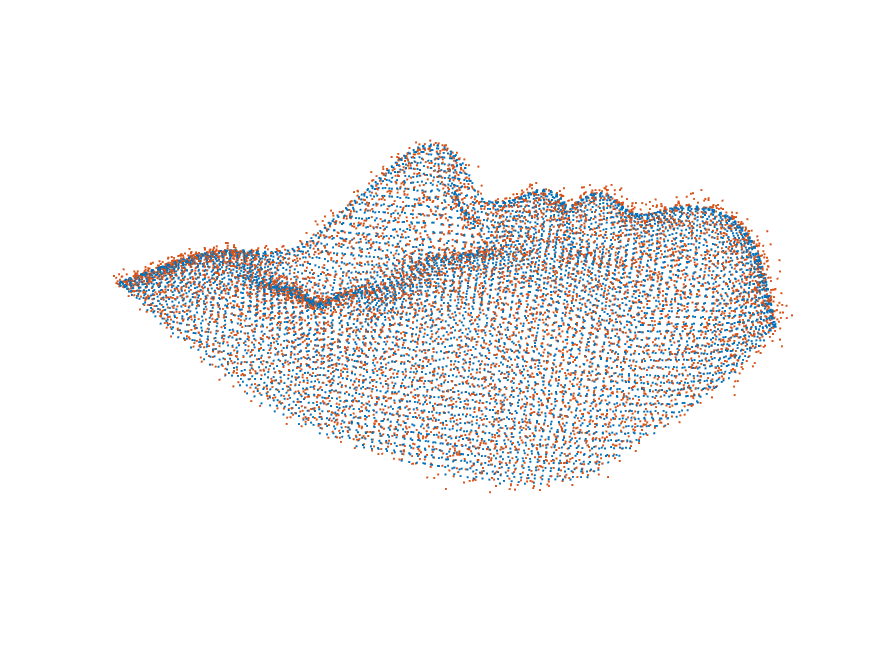}} &
        \includegraphics[trim = 40 75 20 50, clip, height = 1.2in]{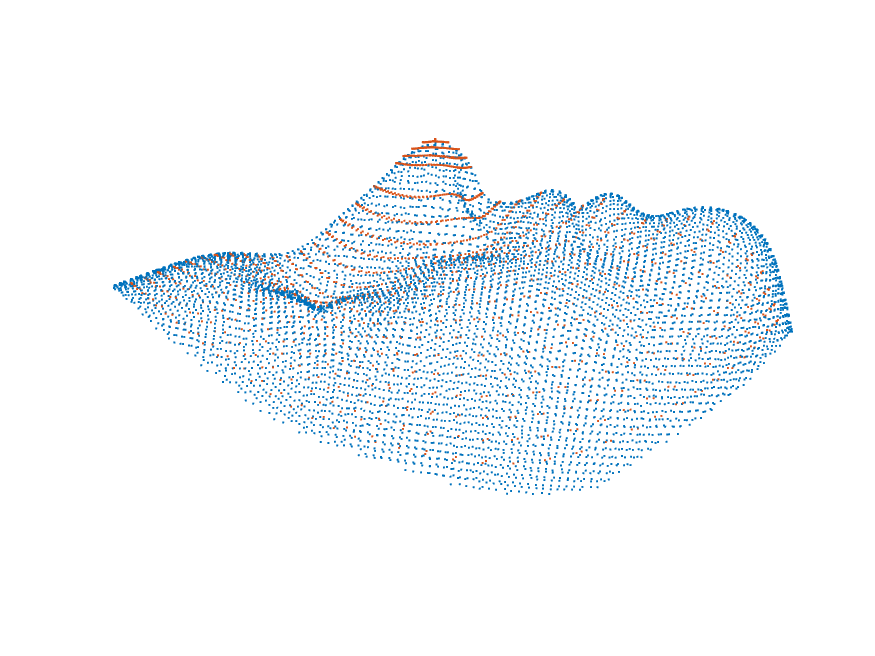}
    \end{tabular}
    \caption{Blue points are the point-wise mean. Left: The red points are the point-wise private mean with $\mu_T=3$. Right: The red points are our private mean with $\mu_T=2.9961$.}
    \label{fig:ErrorFig}
\end{figure}
% \begin{figure}
%     \centering
%     \begin{tabular}{c|c}
%         \includegraphics[trim = 50 20 30 0, clip, height=1.3in]{DiffSameFigAng1mu3}
%         \includegraphics[trim = 100 40 110 20, clip, height=1.3in]{DiffSameFigAng3mu3}
%     \end{tabular}
%     \caption{Caption}
%     \label{fig:enter-label}
% \end{figure}
% \begin{figure}
%     \centering
%     \begin{tabular}{@{}c@{} @{}c@{} }
%         % \includegraphics[trim = 50 20 30 0, clip, height=1.3in]{MeanGDPmu2.492Ang1} &
%         % \includegraphics[trim = 100 40 110 20, clip, height=1.3in]{MeanGDPmu2.492Ang2} \\
%         \includegraphics[trim = 50 20 30 0, clip, height=1.3in]{MeanGDPmu2.9661Ang1} &
%         \includegraphics[trim = 100 40 110 20, clip, height=1.3in]{MeanGDPmu2.9661Ang2} \\
%     \end{tabular}
%     \caption{Two angles of a differentially private face constructed from private mean face radial curves with $\mu_T=2.9961$.}
%     \label{fig:PrivateFaces}
% \end{figure}

%-------------------------------------------------------------------------
\section{Conclusions}
We have developed a framework for releasing a $\mu$-GDP human face with our novel face radial curve representation. We extended $(\epsilon,\delta)$-DP FDA \citep{mirshani2017establishing} techniques into the  $\mu$-GDP framework \citep{dong2019gaussian} to take advantage of its tight composition of budgets. We utilized the new $\mu$-GDP FDA framework specifically for faces but note this is applicable for other FDA applications. Further, we utilize the shape analysis techniques of \citet{samir2006three,drira2010pose,jermyn2017elastic} to create a novel curve representation of a face. We focus on human faces, however, one can use our framework to sanitize other surfaces which are diffeomorphic to a disk.
%such as for hand gestures or sections of the surface of the brain such that we ensure that the surface itself has no holes and is smooth. In practice holes may or may not be an issue as with, for instance, the nose holes of the face are closed. 
Our representation changes the problem of sanitization from needing thousands of point-wise estimations to a few dozen functional estimates. We discuss the limitations of our methods in \ref{ss:LimitBreak}.

%Ones face contains identifying characteristics and its privacy is perhaps clear. However, as noted earlier, given genomic information can be inferred from a human face \cite{venkatesaramani2021re} and the fact that noise can hide that, the importance is more precise.
A persons face contains information of ones identity and, as noted earlier, its image contains genomic information \citep{venkatesaramani2021re}, thus, the need to protect this information while sharing these data (e.g., anthropological studies) is crucial.
We demonstrate via empirical and quantitative results that our methodology adds less noise and preserves the structure of the face. We chose the number of face radial curves needed for our representation, but one could develop a private cross-validation method to do so. 

In \ref{ss:Data} we describe how the data is collected by \cite{sero2019facial}, so here we mention some limitations in our experimental results as well as in \ref{ss:LimitBreak}. In a sense, our data is very clean, and hence our methodology is only tested in this circumstance. All faces in the dataset have a neutral expression, so to handle other expressions would require additional processing or extending the methodology. Expression variation is a complex source of variability which is critical in areas such as face recognition \citep{6101586}. %We do however note cases such as experimental studies (e.g. anthropological) and passport photos, and enforce neutral faces positions and are ideal for our methodology. Further, 
It is not entirely clear how to handle such variability for estimation of an average.
The parametrization requires a genus-0 surface, i.e. no holes, so if a face has missing data this would need to be rectified. For instance, preprocessing by patching ``holes" or missing data has been considered by \cite{passalis2011using} which leverage the symmetry of the face for interpolation. %Facial recognition is an area which encounters similar issues. There are a vast amount of survey papers for face recognition such as \cite{zhao2003face,jafri2009survey,li2020review}, to name a few, all to say these problems are not trivial and
There is a vast amount of survey papers in the face cognition literature such as \cite{zhao2003face,jafri2009survey,li2020review}, to name a few, which encounters similar issues all to say these problems are not trivial and are potential future research opportunities.

%Simplifying the problem from sanitizing the entire face to radial curves comes at the cost of needing to impose a parameterization.
%Parameterizing a face is a non-trivial task which must be done quite carefully. %Further, using the shape analysis tools for optimal registration are computationally expensive. However, the cost we pay for this heavy pre-processing of the data is justified in the simplification of the problem. Effectively we shift the burden of complexity from the privacy mechanism to the processing.

%The question of how many radial curves is important but difficult to answer. Partially the question is philosophical in nature as at some point more radial curves do not ``add" any definition yet they impact the private statistical summary. Determining the optimal number of radial curves can be a bit subjective, however, one could use cross validation or an elbow method to determine when the number of curves plateaus on reconstructing the entire face. We see this as an open question with a vast amount of possible solutions.

%\subsubsection*{Author Contributions}
%

\subsubsection*{Acknowledgments}
The authors would also like to thank the anonymous reviewers which provided constructive feedback. CS would also like to thank Gary Choi for guidance on conformal maps. This work was in part supported by NSF Award No. SES-1853209 to the Pennsylvania State University, and by Huck Institutes of the Life Sciences at the Pennsylvania State University through the Dorothy Foehr Huck and J. Lloyd Huck Chair in Data Privacy and Confidentiality. Content is the responsibility of the authors and does not represent the views of the Huck Institutes. 

\bibliography{ref}

\begin{thebibliography}{54}
\providecommand{\natexlab}[1]{#1}
\providecommand{\url}[1]{\texttt{#1}}
\expandafter\ifx\csname urlstyle\endcsname\relax
  \providecommand{\doi}[1]{doi: #1}\else
  \providecommand{\doi}{doi: \begingroup \urlstyle{rm}\Url}\fi

\bibitem[Alda \& Rubinstein(2017)Alda and Rubinstein]{alda2017bernstein}
Francesco Alda and Benjamin Rubinstein.
\newblock The bernstein mechanism: Function release under differential privacy.
\newblock In \emph{Proceedings of the AAAI Conference on Artificial
  Intelligence}, volume~31, 2017.

\bibitem[Balle \& Wang(2018)Balle and Wang]{balle2018improving}
Borja Balle and Yu-Xiang Wang.
\newblock Improving the gaussian mechanism for differential privacy: Analytical
  calibration and optimal denoising.
\newblock In \emph{International Conference on Machine Learning}, pp.\
  394--403. PMLR, 2018.

\bibitem[Bhattacharya \& Patrangenaru(2003)Bhattacharya and
  Patrangenaru]{bhattacharya2003large}
Rabi Bhattacharya and Vic Patrangenaru.
\newblock Large sample theory of intrinsic and extrinsic sample means on
  manifolds.
\newblock \emph{The Annals of Statistics}, 31\penalty0 (1):\penalty0 1--29,
  2003.

\bibitem[Blum et~al.(2005)Blum, Dwork, McSherry, and Nissim]{blum2005practical}
Avrim Blum, Cynthia Dwork, Frank McSherry, and Kobbi Nissim.
\newblock Practical privacy: the {SuLQ} framework.
\newblock In \emph{Proceedings of the twenty-fourth ACM SIGMOD-SIGACT-SIGART
  symposium on Principles of database systems}, pp.\  128--138, 2005.

\bibitem[Bun \& Steinke(2016)Bun and Steinke]{bun2016concentrated}
Mark Bun and Thomas Steinke.
\newblock Concentrated differential privacy: Simplifications, extensions, and
  lower bounds.
\newblock In \emph{Theory of Cryptography Conference}, pp.\  635--658.
  Springer, 2016.

\bibitem[Chaudhuri et~al.(2013)Chaudhuri, Sarwate, and
  Sinha]{chaudhuri2013near}
Kamalika Chaudhuri, Anand~D Sarwate, and Kaushik Sinha.
\newblock A near-optimal algorithm for differentially-private principal
  components.
\newblock \emph{Journal of Machine Learning Research}, 14, 2013.

\bibitem[Cho et~al.(2019)Cho, Kurtek, and MacEachern]{cho2019aggregated}
Min~Ho Cho, Sebastian Kurtek, and Steven~N MacEachern.
\newblock Aggregated pairwise classification of statistical shapes.
\newblock \emph{arXiv preprint arXiv:1901.07593}, 2019.

\bibitem[Choi(2020)]{ChoiGitHub}
Gary~P.T. Choi.
\newblock disk-conformal-map.
\newblock \url{https://github.com/garyptchoi/disk-conformal-map}, 2020.

\bibitem[Choi et~al.(2020)Choi, Leung-Liu, Gu, and Lui]{choi2020parallelizable}
Gary~PT Choi, Yusan Leung-Liu, Xianfeng Gu, and Lok~Ming Lui.
\newblock Parallelizable global conformal parameterization of simply-connected
  surfaces via partial welding.
\newblock \emph{SIAM Journal on Imaging Sciences}, 13\penalty0 (3):\penalty0
  1049--1083, 2020.

\bibitem[Choi \& Lui(2018)Choi and Lui]{choi2018linear}
Gary Pui-Tung Choi and Lok~Ming Lui.
\newblock A linear formulation for disk conformal parameterization of
  simply-connected open surfaces.
\newblock \emph{Advances in Computational Mathematics}, 44:\penalty0 87--114,
  2018.

\bibitem[Choi \& Lui(2015)Choi and Lui]{choi2015fast}
Pui~Tung Choi and Lok~Ming Lui.
\newblock Fast disk conformal parameterization of simply-connected open
  surfaces.
\newblock \emph{Journal of Scientific Computing}, 65\penalty0 (3):\penalty0
  1065--1090, 2015.

\bibitem[Cignoni et~al.(2008)Cignoni, Callieri, Corsini, Dellepiane, Ganovelli,
  and Ranzuglia]{LocalChapterEvents:ItalChap:ItalianChapConf2008:129-136}
Paolo Cignoni, Marco Callieri, Massimiliano Corsini, Matteo Dellepiane, Fabio
  Ganovelli, and Guido Ranzuglia.
\newblock {MeshLab: an Open-Source Mesh Processing Tool}.
\newblock In Vittorio Scarano, Rosario~De Chiara, and Ugo Erra (eds.),
  \emph{Eurographics Italian Chapter Conference}. The Eurographics Association,
  2008.
\newblock ISBN 978-3-905673-68-5.
\newblock
  \doi{10.2312/LocalChapterEvents/ItalChap/ItalianChapConf2008/129-136}.

\bibitem[Dong et~al.(2019)Dong, Roth, and Su]{dong2019gaussian}
Jinshuo Dong, Aaron Roth, and Weijie~J Su.
\newblock Gaussian differential privacy.
\newblock \emph{arXiv preprint arXiv:1905.02383}, 2019.

\bibitem[Drira et~al.(2010)Drira, Amor, Daoudi, and Srivastava]{drira2010pose}
Hassen Drira, Boulbaba~Ben Amor, Mohamed Daoudi, and Anuj Srivastava.
\newblock Pose and expression-invariant 3d face recognition using elastic
  radial curves.
\newblock In \emph{British machine vision conference}, pp.\  1--11, 2010.

\bibitem[Dwork \& Roth(2014)Dwork and Roth]{dwork2014algorithmic}
Cynthia Dwork and Aaron Roth.
\newblock The algorithmic foundations of differential privacy.
\newblock \emph{Foundations and Trends{\textregistered} in Theoretical Computer
  Science}, 9\penalty0 (3--4):\penalty0 211--407, 2014.

\bibitem[Dwork et~al.(2006)Dwork, McSherry, Nissim, and
  Smith]{dwork2006calibrating}
Cynthia Dwork, Frank McSherry, Kobbi Nissim, and Adam Smith.
\newblock Calibrating noise to sensitivity in private data analysis.
\newblock In \emph{Theory of cryptography conference}, pp.\  265--284.
  Springer, 2006.

\bibitem[Gneiting(2013)]{gneiting2013strictly}
Tilmann Gneiting.
\newblock Strictly and non-strictly positive definite functions on spheres.
\newblock \emph{Bernoulli}, 19(4):\penalty0 1327--1349, 2013.

\bibitem[Hall et~al.(2013)Hall, Rinaldo, and Wasserman]{hall2013differential}
Rob Hall, Alessandro Rinaldo, and Larry Wasserman.
\newblock Differential privacy for functions and functional data.
\newblock \emph{The Journal of Machine Learning Research}, 14\penalty0
  (1):\penalty0 703--727, 2013.

\bibitem[Han et~al.(2022)Han, Mishra, Jawanpuria, and
  Gao]{han2022differentially}
Andi Han, Bamdev Mishra, Pratik Jawanpuria, and Junbin Gao.
\newblock Differentially private riemannian optimization.
\newblock \emph{arXiv preprint arXiv:2205.09494}, 2022.

\bibitem[Jafri \& Arabnia(2009)Jafri and Arabnia]{jafri2009survey}
Rabia Jafri and Hamid~R Arabnia.
\newblock A survey of face recognition techniques.
\newblock \emph{journal of information processing systems}, 5\penalty0
  (2):\penalty0 41--68, 2009.

\bibitem[Jermyn et~al.(2017)Jermyn, Kurtek, Laga, and
  Srivastava]{jermyn2017elastic}
Ian~H Jermyn, Sebastian Kurtek, Hamid Laga, and Anuj Srivastava.
\newblock \emph{Elastic shape analysis of three-dimensional objects}.
\newblock Springer, 2017.

\bibitem[Jiang et~al.(2023)Jiang, Chang, Liu, Ding, Kong, and
  Jiang]{jiang2023gaussian}
Yangdi Jiang, Xiaotian Chang, Yi~Liu, Lei Ding, Linglong Kong, and Bei Jiang.
\newblock Gaussian differential privacy on riemannian manifolds.
\newblock \emph{arXiv preprint arXiv:2311.10101}, 2023.

\bibitem[Kendall(1984)]{kendall1984shape}
David~G Kendall.
\newblock Shape manifolds, procrustean metrics, and complex projective spaces.
\newblock \emph{Bulletin of the London mathematical society}, 16\penalty0
  (2):\penalty0 81--121, 1984.

\bibitem[Klassen et~al.(2004)Klassen, Srivastava, Mio, and
  Joshi]{klassen2004analysis}
Eric Klassen, Anuj Srivastava, M~Mio, and Shantanu~H Joshi.
\newblock Analysis of planar shapes using geodesic paths on shape spaces.
\newblock \emph{IEEE transactions on pattern analysis and machine
  intelligence}, 26\penalty0 (3):\penalty0 372--383, 2004.

\bibitem[Klimentidis \& Shriver(2009)Klimentidis and
  Shriver]{klimentidis2009estimating}
Yann~C Klimentidis and Mark~D Shriver.
\newblock Estimating genetic ancestry proportions from faces.
\newblock \emph{PLoS One}, 4\penalty0 (2):\penalty0 e4460, 2009.

\bibitem[Kortli et~al.(2020)Kortli, Jridi, Al~Falou, and Atri]{kortli2020face}
Yassin Kortli, Maher Jridi, Ayman Al~Falou, and Mohamed Atri.
\newblock Face recognition systems: A survey.
\newblock \emph{Sensors}, 20\penalty0 (2):\penalty0 342, 2020.

\bibitem[Laga(2022)]{LagaGitHub}
Hamid Laga.
\newblock Srnf.
\newblock \url{https://github.com/hamidlaga/SRNF}, 2022.

\bibitem[Laga et~al.(2018)Laga, Guo, Tabia, Fisher, and Bennamoun]{laga20183d}
Hamid Laga, Yulan Guo, Hedi Tabia, Robert~B Fisher, and Mohammed Bennamoun.
\newblock \emph{3D Shape analysis: fundamentals, theory, and applications}.
\newblock John Wiley \& Sons, 2018.

\bibitem[Li et~al.(2020)Li, Mu, Li, and Peng]{li2020review}
Lixiang Li, Xiaohui Mu, Siying Li, and Haipeng Peng.
\newblock A review of face recognition technology.
\newblock \emph{IEEE access}, 8:\penalty0 139110--139120, 2020.

\bibitem[Li \& Choi(2021)Li and Choi]{li2021deepblur}
Tao Li and Min~Soo Choi.
\newblock Deepblur: A simple and effective method for natural image
  obfuscation.
\newblock \emph{arXiv preprint arXiv:2104.02655}, 1:\penalty0 3, 2021.

\bibitem[Mironov(2017)]{mironov2017renyi}
Ilya Mironov.
\newblock R{\'e}nyi differential privacy.
\newblock In \emph{2017 IEEE 30th computer security foundations symposium
  (CSF)}, pp.\  263--275. IEEE, 2017.

\bibitem[Mirshani et~al.(2017)Mirshani, Reimherr, and
  Slavkovic]{mirshani2017establishing}
Ardalan Mirshani, Matthew Reimherr, and A~Slavkovic.
\newblock Establishing statistical privacy for functional data via functional
  densities.
\newblock \emph{arXiv preprint arXiv:1711.06660}, 2017.

\bibitem[Naqvi et~al.(2021)Naqvi, Sleyp, Hoskens, Indencleef, Spence,
  Bruffaerts, Radwan, Eller, Richmond, Shriver, et~al.]{naqvi2021shared}
Sahin Naqvi, Yoeri Sleyp, Hanne Hoskens, Karlijne Indencleef, Jeffrey~P Spence,
  Rose Bruffaerts, Ahmed Radwan, Ryan~J Eller, Stephen Richmond, Mark~D
  Shriver, et~al.
\newblock Shared heritability of human face and brain shape.
\newblock \emph{Nature genetics}, 53\penalty0 (6):\penalty0 830--839, 2021.

\bibitem[Passalis et~al.(2011)Passalis, Perakis, Theoharis, and
  Kakadiaris]{passalis2011using}
Georgios Passalis, Panagiotis Perakis, Theoharis Theoharis, and Ioannis~A
  Kakadiaris.
\newblock Using facial symmetry to handle pose variations in real-world 3d face
  recognition.
\newblock \emph{IEEE Transactions on Pattern Analysis and Machine
  Intelligence}, 33\penalty0 (10):\penalty0 1938--1951, 2011.

\bibitem[Pennec(2006)]{pennec2006intrinsic}
Xavier Pennec.
\newblock Intrinsic statistics on riemannian manifolds: Basic tools for
  geometric measurements.
\newblock \emph{Journal of Mathematical Imaging and Vision}, 25:\penalty0
  127--154, 2006.

\bibitem[Reimherr et~al.(2021)Reimherr, Bharath, and
  Soto]{NEURIPS2021_6600e06f}
Matthew Reimherr, Karthik Bharath, and Carlos Soto.
\newblock Differential privacy over riemannian manifolds.
\newblock In M.~Ranzato, A.~Beygelzimer, Y.~Dauphin, P.S. Liang, and J.~Wortman
  Vaughan (eds.), \emph{Advances in Neural Information Processing Systems},
  volume~34, pp.\  12292--12303. Curran Associates, Inc., 2021.
\newblock URL
  \url{https://proceedings.neurips.cc/paper/2021/file/6600e06fe9350b62c1e343504d4a7b86-Paper.pdf}.

\bibitem[Samir et~al.(2006)Samir, Srivastava, and Daoudi]{samir2006three}
Chafik Samir, Anuj Srivastava, and Mohamed Daoudi.
\newblock Three-dimensional face recognition using shapes of facial curves.
\newblock \emph{IEEE Transactions on Pattern Analysis and Machine
  Intelligence}, 28\penalty0 (11):\penalty0 1858--1863, 2006.

\bibitem[Sero et~al.(2019)Sero, Zaidi, Li, White, Zarzar, Marazita, Weinberg,
  Suetens, Vandermeulen, Wagner, et~al.]{sero2019facial}
Dzemila Sero, Arslan Zaidi, Jiarui Li, Julie~D White, Tom{\'a}s B~Gonz{\'a}lez
  Zarzar, Mary~L Marazita, Seth~M Weinberg, Paul Suetens, Dirk Vandermeulen,
  Jennifer~K Wagner, et~al.
\newblock Facial recognition from dna using face-to-dna classifiers.
\newblock \emph{Nature communications}, 10\penalty0 (1):\penalty0 2557, 2019.

\bibitem[Sheffet(2019)]{sheffet2019old}
Or~Sheffet.
\newblock Old techniques in differentially private linear regression.
\newblock In \emph{Algorithmic Learning Theory}, pp.\  789--827. PMLR, 2019.

\bibitem[Smeets et~al.(2012)Smeets, Claes, Hermans, Vandermeulen, and
  Suetens]{6101586}
Dirk Smeets, Peter Claes, Jeroen Hermans, Dirk Vandermeulen, and Paul Suetens.
\newblock A comparative study of 3-d face recognition under expression
  variations.
\newblock \emph{IEEE Transactions on Systems, Man, and Cybernetics, Part C
  (Applications and Reviews)}, 42\penalty0 (5):\penalty0 710--727, 2012.
\newblock \doi{10.1109/TSMCC.2011.2174221}.

\bibitem[Soto et~al.(2022)Soto, Bharath, Reimherr, and
  Slavkovi{\'c}]{soto2022shape}
Carlos Soto, Karthik Bharath, Matthew Reimherr, and Aleksandra Slavkovi{\'c}.
\newblock Shape and structure preserving differential privacy.
\newblock \emph{Advances in Neural Information Processing Systems},
  35:\penalty0 24693--24705, 2022.

\bibitem[Srivastava et~al.(2010)Srivastava, Klassen, Joshi, and
  Jermyn]{srivastava2010shape}
Anuj Srivastava, Eric Klassen, Shantanu~H Joshi, and Ian~H Jermyn.
\newblock Shape analysis of elastic curves in euclidean spaces.
\newblock \emph{IEEE transactions on pattern analysis and machine
  intelligence}, 33\penalty0 (7):\penalty0 1415--1428, 2010.

\bibitem[Su et~al.(2020)Su, Bauer, Preston, Laga, and Klassen]{su2020shape}
Zhe Su, Martin Bauer, Stephen~C Preston, Hamid Laga, and Eric Klassen.
\newblock Shape analysis of surfaces using general elastic metrics.
\newblock \emph{Journal of Mathematical Imaging and Vision}, 62:\penalty0
  1087--1106, 2020.

\bibitem[Sun et~al.(2022)Sun, Zhou, Wang, Yu, Jin, and Zhou]{sun20223d}
Wenyuan Sun, Ping Zhou, Yangang Wang, Zongpu Yu, Jing Jin, and Guangquan Zhou.
\newblock 3d face parsing via surface parameterization and 2d semantic
  segmentation network.
\newblock \emph{arXiv preprint arXiv:2206.09221}, 2022.

\bibitem[Tokola et~al.(2015)Tokola, Mikkilineni, and Boehnen]{tokola20153d}
Ryan Tokola, Aravind Mikkilineni, and Christopher Boehnen.
\newblock 3d face analysis for demographic biometrics.
\newblock In \emph{2015 International Conference on Biometrics (ICB)}, pp.\
  201--207. IEEE, 2015.

\bibitem[Trouv{\'e} \& Younes(2005)Trouv{\'e} and Younes]{trouve2005local}
Alain Trouv{\'e} and Laurent Younes.
\newblock Local geometry of deformable templates.
\newblock \emph{SIAM journal on mathematical analysis}, 37\penalty0
  (1):\penalty0 17--59, 2005.

\bibitem[Utpala et~al.(2022)Utpala, Vepakomma, and
  Miolane]{utpala2022differentially}
Saiteja Utpala, Praneeth Vepakomma, and Nina Miolane.
\newblock Differentially private fr$\backslash$'echet mean on the manifold of
  symmetric positive definite (spd) matrices.
\newblock \emph{arXiv preprint arXiv:2208.04245}, 2022.

\bibitem[Venkatesaramani et~al.(2021)Venkatesaramani, Malin, and
  Vorobeychik]{venkatesaramani2021re}
Rajagopal Venkatesaramani, Bradley~A Malin, and Yevgeniy Vorobeychik.
\newblock Re-identification of individuals in genomic datasets using public
  face images.
\newblock \emph{Science advances}, 7\penalty0 (47):\penalty0 eabg3296, 2021.

\bibitem[Vishwamitra et~al.(2017)Vishwamitra, Knijnenburg, Hu, Kelly~Caine,
  et~al.]{vishwamitra2017blur}
Nishant Vishwamitra, Bart Knijnenburg, Hongxin Hu, Yifang~P Kelly~Caine, et~al.
\newblock Blur vs. block: Investigating the effectiveness of privacy-enhancing
  obfuscation for images.
\newblock In \emph{Proceedings of the IEEE Conference on Computer Vision and
  Pattern Recognition Workshops}, pp.\  39--47, 2017.

\bibitem[Wallace et~al.(2014)Wallace, Srivastava, Telu, and
  Sim{\'o}n-Manso]{wallace2014pairwise}
William~E Wallace, Anuj Srivastava, Kelly~H Telu, and Y~Sim{\'o}n-Manso.
\newblock Pairwise alignment of chromatograms using an extended fisher--rao
  metric.
\newblock \emph{Analytica Chimica Acta}, 841:\penalty0 10--16, 2014.

\bibitem[Wasserman \& Zhou(2010)Wasserman and Zhou]{wasserman2010statistical}
Larry Wasserman and Shuheng Zhou.
\newblock A statistical framework for differential privacy.
\newblock \emph{Journal of the American Statistical Association}, 105\penalty0
  (489):\penalty0 375--389, 2010.

\bibitem[Weinberg et~al.(2019)Weinberg, Roosenboom, Shaffer, Shriver, Wysocka,
  and Claes]{weinberg2019hunting}
Seth~M Weinberg, Jasmien Roosenboom, John~R Shaffer, Mark~D Shriver, Joanna
  Wysocka, and Peter Claes.
\newblock Hunting for genes that shape human faces: Initial successes and
  challenges for the future.
\newblock \emph{Orthodontics \& craniofacial research}, 22:\penalty0 207--212,
  2019.

\bibitem[White et~al.(2021)White, Indencleef, Naqvi, Eller, Hoskens,
  Roosenboom, Lee, Li, Mohammed, Richmond, et~al.]{white2021insights}
Julie~D White, Karlijne Indencleef, Sahin Naqvi, Ryan~J Eller, Hanne Hoskens,
  Jasmien Roosenboom, Myoung~Keun Lee, Jiarui Li, Jaaved Mohammed, Stephen
  Richmond, et~al.
\newblock Insights into the genetic architecture of the human face.
\newblock \emph{Nature genetics}, 53\penalty0 (1):\penalty0 45--53, 2021.

\bibitem[Zhao et~al.(2003)Zhao, Chellappa, Phillips, and
  Rosenfeld]{zhao2003face}
Wenyi Zhao, Rama Chellappa, P~Jonathon Phillips, and Azriel Rosenfeld.
\newblock Face recognition: A literature survey.
\newblock \emph{ACM computing surveys (CSUR)}, 35\penalty0 (4):\penalty0
  399--458, 2003.

\end{thebibliography}
\bibliographystyle{iclr2025_conference}

\appendix
\section{Appendix}

%\section{Supplemental Material}
%------------------------------------------------------------------
\subsection{Limitations}\label{ss:LimitBreak}
We do not collect any data, however, in \ref{ss:Data} we describe the data and how it was collected. In many regards our data is clean, so we did not test our methods on noisy data. In a sense, we can consider our data as processed data, so noisy data would require some processing which we do not consider in this work. It is possible that if the data is more noisy then it would require further smoothing which our method does allow. We further have an assumption of a genus-0 surface (a face with no holes) which may be an issue if there are missing data points. This type of data would require additional pre-processing. We also only test our method on one dataset; one could encounter the issue of variability of noise in other datasets but this can be overcome with the smoothness parameters in our method. 

A limitation which we have already mentioned in \ref{ss:ExRes} is the method for computing sensitivity, $\Delta$, in both our method and the point-wise method which we compare ourselves to. This, however, is an inherent issue in differential privacy rather than our method. That is, for both our method and the point-wise private estimate, we use the data to estimate the sensitivity. Ideally one could use a second dataset or a training set to estimate the many sensitivity values needed in a private way.

In terms of limitations of privacy and fairness, these are the exact problems we address. As mentioned in \ref{s:intro}, one may want to release private average faces for demographics. Our methods are intended for facilitate this goal although we do not consider variability of sample sizes by demographics.
%------------------------------------------------------------------
\subsection{Experiments Compute Resources}\label{ss:Compute}
We ran experiments almost exclusively on a desktop computer with an Intel i7 processor and 32GB of RAM on Windows 11. Creating a triangulated mesh for a face takes about 5 seconds on the desktop, so this was done locally for each face. Parameterizing each face costs about 20 seconds on the desktop, this was also done locally. Reparameterizing, however, each face to a template costs upwards of 3 minutes on the desktop, so since we have 1000 faces we did this on a supercomputer. 

Computing the mean and privatization with our method costs about 3 total minutes on the desktop computer. The mean computation and privatization of the point-wise method only requires about 1 minute of computational time on the desktop.

%------------------------------------------------------------------
\subsection{Reparameterization}\label{ss:reg}
To optimally register two surfaces we use the methods as described in \citet{jermyn2017elastic} which we summarize next.
Suppose we have two surfaces $f_1,f_2\in \mcF$  which we wish to optimally register over the parameterization group $\Gamma$.
Here $\mcF$ is as in \ref{ss:esa} $\mcF\ni f:D\rightarrow \mbR^3$ where $D$ is the unit disk.
The action of $\Gamma$ on a surface is right composition, $f\circ\gamma$ for $\gamma\in\Gamma$, and reparameterizes the surfaces but does not change its image which in our case is the face. %To do this a gradient descent approach is used. 
The authors in \citet{jermyn2017elastic} leverage a transformation of $f$ referred to as the square-root normal field (SRNF) defined as $q=\frac{n}{|n|^{1/2}}$ where $n$ is the normal vector $n=\frac{\partial f}{\partial u}\times\frac{\partial f}{\partial v}$. The corresponding action of $\Gamma$ on $q$ is then $(q,\gamma):=\sqrt{J(\gamma)}q\circ\gamma$. This transformation is necessary for an isometric action of $\Gamma$ but the details of this are lengthy and not necessary for our application.

The authors define an energy function $E_{reg}:\Gamma\rightarrow\mbR_{\geq0}$ to implement an iterative gradient descent method for the registration. The energy is defined as $E_{reg}(\gamma) = \|q_1-(\tilde{q}_2,\gamma)\|^2$ where $\tilde{q}_2$ is the ``current" stage of the surface being reparameterized.
Here $\gamma$ %and $\gamma_0$ denote the 
is the incremental reparameterization with
%and the current iteration of reparameterization  respectively, so 
$\tilde{q}_2=({q}_2,\gamma)$.
Here $q_1$ would be our template surface from \ref{ss:esa} and we register all surfaces to this template. 

Since the reparameterization is done in an \textit{iterative} manner it results in incrementally improved registration. This method is iterative as the gradient is taken about the identity of $\Gamma$, $\gamma_{identity}$ and thus lives in the tangent space of this element.
Suppose we have an orthonormal basis $\mcB=\{b\}$ for the tangent space of the $\Gamma$ at the identity, $T_{id}(\Gamma)$ where each $b$ is a unit ``vector" in the tangent space $T_{id}(\Gamma)$.
%The directional derivative of $E_{reg}$ in the direction of $b$ is given by $\langle q_1 -\phi(\gamma_{id}),d\phi(b)\rangle_2 b$.  Assuming we have such a basis, 
The full gradient is given by %$\sum_{b_i\in\mcB}\langle q_1 -\phi(\gamma_{id}),d\phi(b_i)\rangle_2 b_i$
$\sum_{b_i\in\mcB}\langle q_1 -\tilde{q}_2,d  (b_i,\gamma) \rangle_2 b_i$ and we use the implementation mentioned in \ref{ss:esa} and refer the interested reader to the cited materials.

This reparameterization is costly due to the dimension, iterative construction, and reliance on a basis. To alleviate some expense, we set the center of the surface as the tip of the nose as noted in \ref{ss:Mobius} with the triangulated mesh before we impose the disk parameterization. The reparameterization iterative method thus has to search over a smaller space as the center of the disk is pre-registered. 
  % $E_{reg}(\gamma) = \|q_1-(\tilde{q}_2,\gamma)\|^2:=\|q_1-\phi(\gamma)\|^2$.
  
% for $\gamma(s)=s$. 

%------------------------------------------------------------------
\subsection{Data Details}\label{ss:Data}
For a full description of how the data is collected we refer to \cite{sero2019facial} but we summarize and emphasize the relevant points. All participants have the same neutral face expressions, so the data is not heterogeneous in terms of facial expression. Each face is captured with 7150 points which the authors refer to as quasi-landmarks. Further, these landmarks are registered across all individuals which \cite{sero2019facial} refer to as ``homologous." This registration is a point-wise correspondence across faces while the registration in \ref{ss:reg} is an entire disk correspondence. %The set of faces here is the collection $\{X_i\}$ with each $X_i\in\mbR^{7150\times 3}$, 
Similar to our approach the faces are aligned using Procrustes Analysis, which finds the optimal rotation but again considering the faces as a set of points not a disk-like surface.

% 3D image registration and quality control. The 3D surface images and their 
% reflections were registered using the MeshMonk registration framework (v.0.0.6)24 
% in Matlab 2017b. This process results in a homologous configuration of 7,160 
% spatially dense quasi-landmarks, allowing the image data from different individuals 
% and camera systems to be standardized24. Images differing greatly from the norm 
% or with large holes were investigated manually and either removed or re-processed, 
% with details available in the Supplementary Methods. Although variation in 
% asymmetric facial features is of interest, in this work we sought only to study 
% variation in symmetric facial shape.
%------------------------------------------------------------------
\subsection{Supplemental Notes on DP}\label{ss:DPextra}
We first present the definition for approximate differential privacy.
\begin{definition}[\citep{dwork2006calibrating}]
Let $D\sim D'$ and $\tilde{h}(D)$ be a random privacy mechanism, the mechanism is said to achieve approximate differential privacy, $(\epsilon,\delta)$-DP, for some  $\epsilon>0$ and $0<\delta<1$, if it satisfies the probabilistic inequality
$$P(\tilde{h}(D)\in A)\leq e^\epsilon P(\tilde{h}(D')\in A)+\delta,$$
for any measurable set $A$.
\end{definition}

Both $\epsilon$ and $\delta$ are pre-specified parameters referred to as the privacy budget. When $\delta=0$, this is referred to as \textit{pure} differential privacy. Differential privacy is an attribute of the random mechanism and roughly states that the distributions over adjacent datasets are not too different. To achieve approximate DP for functional data \cite{mirshani2017establishing} establish the following mechanism.
\begin{theorem}[\citet{mirshani2017establishing}]\label{thm:fdp}
    Let $h(D)$ be a functional summary of a dataset $D$ that is compatible with standard Gaussian process noise $Z$ and $\epsilon\leq 1$. We have that $\tilde{h}(D):=h(D)+\sigma Z$ achieves $(\epsilon,\delta)-$DP over $\mbH$ where $\sigma \geq \frac{2\log(2/\delta)}{\epsilon^2}\Delta^2 $. Here $\Delta$ is the \textit{global sensitivity} of the summary $h(D)$ and $\Delta^2=\text{sup}_{D\sim D'}\|h(D)-h(D')\|_{\mcH}^2$ where the norm is over the %Cameron-Martin 
space of the noise $Z$.
\end{theorem}

Assuming compatibility, let $Q$ denote the probability measure induced by $Z$ over $\mbH$ and $\{P(D):D\in\mcD\}$ denote the family of measures over $\mbH$ induced by $\tilde{h}(D)$ as in Theorem \ref{thm:fdp}. The density of $\tilde{h}(D)$ over $\mbH$ with respect to $Q$, which exists if compatability holds, takes the form 
$$\frac{dP(D)}{dQ}(y)=\exp\left\{ -\frac{1}{2\sigma}\left[\|h(D)\|_{\mcH}^2-2T_{h(D)}(y)\right] \right\},$$ 
$Q$ almost everywhere with $T_{h(D)}(y)=\langle h(D),y\rangle_\mcH$.
The inner product on $\mcH$ can be defined in terms of eigenvalues ($\lambda_i$) and eigenfunctions ($b_i$), as $\langle x,y\rangle_\mcH=\sum_i\lambda_i^{-1} \langle x,b_i\rangle_{\mbH}\langle y,b_i\rangle_{\mbH}$ but can generally be expressed using any basis.
%------------------------------------------------------------------
\subsection{Private Point-wise Face}
Here we have more details on the benchmark method. We have $n=1000$ faces each with 7150 points. Let $D=\{X_i\}$ and $X_i\in\mbR^{7150\times 3}$ be the set of faces; further let $X_i[k,l]$ be the $lth$ coordinate, ($x,y,z$), of the $kth$ point of the $ith$ face, $l\in\{1,2,3\}$, $k\in\{1,\dots,7150\}$, and $i\in\{1,\dots, 1000\}$. The points are ``registered" across all the faces i.e. $X_i[k,\cdot]$ and $X_j[k,\cdot]$ represent the same facial feature e.g. the tip of the nose.

We compute the point-wise mean of $D$, the faces, as $\bar{X}=1/n\sum_i X_i$ and display this in the right panel of Figure \ref{fig:meanfaces}. We have a total budget of $\mu_T$ for the entire face which we need to split across all points, $\bar{X}[k,\cdot]$. We sanitize, independently, the 7150 points, $\bar{X}[k,\cdot]$, each of which has 3 coordinates $\bar{X}[k,l]$. To split the budget evenly, each point's coordinate is allocated $\mu_p:=\sqrt{\mu_T^2/(7150\cdot 3)}$ of the budget. That is, we require $(7150\cdot 3)$ many mechanisms and their budget composition is $\sqrt{\sum \mu_p^2}$, since we are using $\mu$-GDP. For each coordinate of each point of the average, $\bar{X}[k,l]$, we compute sensitivity as $\Delta_{k,l}=\max_{i,j}|X_i[k,l]-X_j[k,l]|$. We add noise to each coordinate of each point as $\bar{X}[k,l]+\Delta_{k,l}/\mu_p \cdot z$ where $z\sim \mcN(0,1)$. 
%------------------------------------------------------------------
\subsection{M\"obius Transformation}\label{ss:Mobius}
The reparameterization defined in \ref{ss:Param} is computationally expensive. To make the repameterization less expensive, we leverage the M\"obius transformation of conformal maps. We give a high level idea of this transformation with more details available in \cite{choi2015fast,choi2018linear,choi2020parallelizable} and implementation available at GitHub repository \citet{ChoiGitHub}.

The left most panel of Figure \ref{fig:ConfMap} displays a triangulated mesh of a face. The next step in our pre-processing is to generate a disk conformal map. The middle column of Figure \ref{fig:ConfMap} displays two different disk conformal maps of the same triangulated mesh in the left panel. It is, admittedly, difficult to discern features, but in each panel of the middle column one can see approximately four dense areas of which the middle corresponds to the nose. Recall that this disk conformal map embeds the triangulated mesh onto the disk while attempting to locally preserve all angles. 

The panels in the right column of Figure \ref{fig:ConfMap} are the disk parameterized surface corresponding to the adjacent middle column disk conformal map. The center of the disk conformal disk, and hence the disk parameterized surface, is not known a priori. Either of these disk parameterized surfaces will suffice in our construction as we can optimally register it to a template as explained in \ref{ss:esa}. However,  since reparameterization is costly, we can save cost by prespecifying the center of the disk of both the template and each face. A natural choice is to designate the center of the disk with feature such as the tip of the nose.

The bottom panel of the middle row is constructed using the M\"obius transformation on the disk conformal map of the top panel of the middle row. We note that the M\"obius transformation is for the disk conformal map and not the disk paramterized surface. This M\"obius transformation costs effectively no time, but reparameterization of the bottom right disk parameterized surface is much faster than reparameterizing the top right to the template. This is because the methodology in \ref{ss:reg} looks at the gradient about the identity, so shifting the center of the disk requires a lot of energy. Having prespecified the nose as the center of the disk, though, the search space for reparamterization is much smaller. We lastly also note that all from the left and right columns of the figure have the same shape but have different representations; that is to say, shape is not effected by its parameterization nor representation.

\begin{figure}
    \centering
    \begin{tabular}{@{}c@{}@{}c@{}}    
        \begin{tabular}{c}
            \fbox{\includegraphics[trim = 20 50 20 30, clip, height = 1.1in]{TriangulationV10010.png}}  
        \end{tabular} &
    
        \begin{tabular}{@{}c@{} @{}c@{} }       
            \fbox{\includegraphics[trim = 85 0 60 0, clip, height = 1.5in]{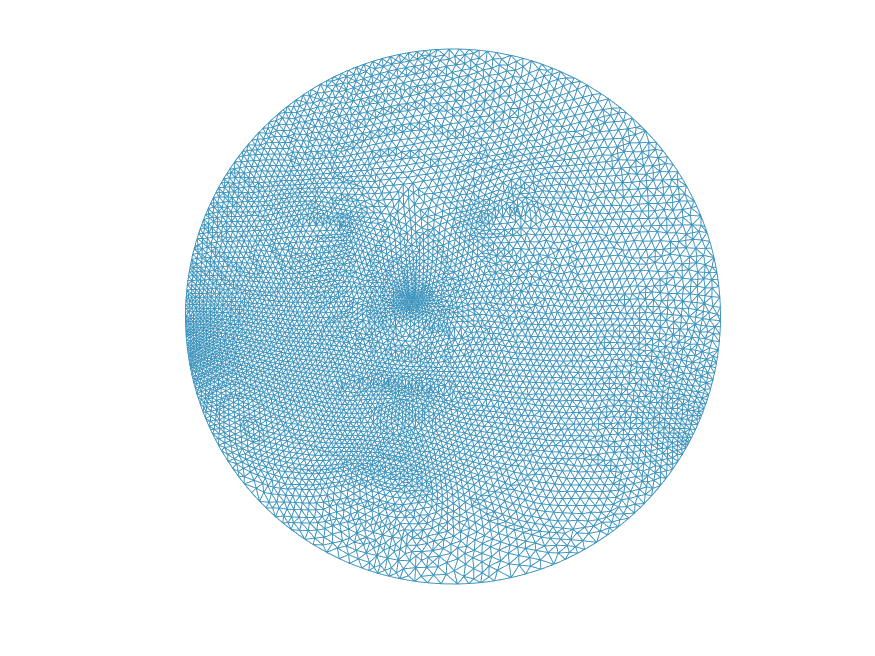}} & 
            \fbox{\includegraphics[trim = 20 50 20 30, clip, height = 1.1in]{ParamFace0010.png}} \\    
            
            \fbox{\includegraphics[trim = 85 0 60 0, clip, height = 1.5in]{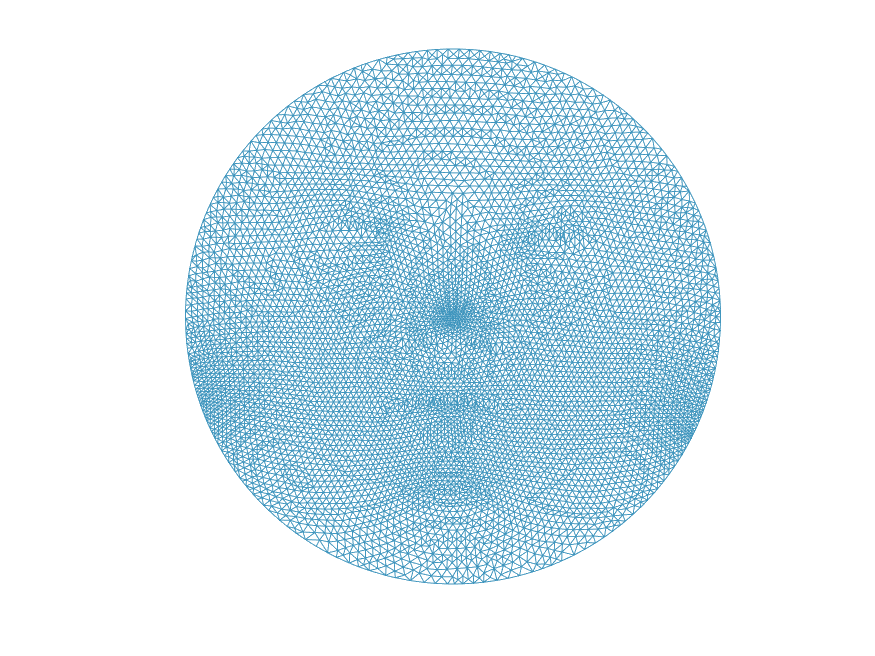}} &    
            \fbox{\includegraphics[trim = 20 50 20 30, clip, height = 1.1in]{ParamFaceMobius0010.png}} 
        \end{tabular}
    \end{tabular}

    \caption{Left column: A triangulated mesh of a face. Middle column: Two disk conformal maps of the triangulated mesh on the left. Right column: The disk parameterized surface resulting from the adjacent middle column conformal map.}
    \label{fig:ConfMap}
\end{figure}

% We re-specify the center of each mapping with the M\"obius transformation \cite{choi2020parallelizable} which modifies a conformal mapping, thus setting the center of the disk as the tip of the nose. The right panel of Figure \ref{fig:ConfMap} displays the transformed conformal mapping of the corresponding left panel.
%------------------------------------------------------------------
\subsection{Proof Details}\label{ss:proof}

We include details to the proof of Theorem \ref{thm:FDAGDP}. First we show the bound on the upper bound on the privacy loss random variable.
\begin{align*}
    P(PL\geq \epsilon) &= P\left(\log\left[\frac{\exp\{-\frac{1}{2\sigma}(\|h(D)\|_{\mcH}^2-2T_D(x))\}}{\exp\{-\frac{1}{2\sigma}(\|h(D')\|_{\mcH}^2-2T_{D'}(x))\}}\right]\geq \epsilon\right)\\
    &= P\left(\log\left[        
    \exp\{-\frac{1}{2\sigma}       
    \left(\|h(D)\|_{\mcH}^2-\|h(D')\|_{\mcH}^2        
    -2(T_D-T_{D'})(x)\right) \}       
    \right]\geq \epsilon\right)\\
     &= P\left(       
    -\frac{1}{2\sigma^2}       
    (\|h(D)\|_{\mcH}^2-\|h(D')\|_{\mcH}^2        
    -2(T_D-T_{D'})(x))      
    \geq \epsilon\right)\\
\end{align*}
Note that we have $\|h(D')\|_{\mcH}-\|h(D)\|_{\mcH}=\|h(D)-h(D')\|_{\mcH}-2\langle h(D)-h(D'),h(D)\rangle_{\mcH}$ and $\langle h(D)-h(D'),h(D)\rangle_{\mcH}=(T_D-T_{D'})(h(D))$. Thus it follows that,
\begin{align*}
    & P\left(-\frac{1}{2\sigma^2}(\|h(D)\|_{\mcH}^2-\|h(D')\|_{\mcH}^2-2(T_D-T_{D'})(x)) \geq \epsilon\right) \\&= 
        P\left(-\frac{1}{2\sigma^2}\left(-\|h(D)-h(D')\|_{\mcH}^2 -2(T_D-T_{D'})(x-h(D)\right) \geq \epsilon\right) \\
    &\leq  P\left(-\frac{1}{2\sigma^2}\left(-\Delta^2 -2(T_D-T_{D'})(x-h(D))\right) \geq \epsilon\right)\\
    &=  P\left( (T_D-T_{D'})(x-h(D)) \geq \sigma^2\epsilon-\frac{\Delta^2}{2}\right) \\
    &=  P\left( \sigma\Delta Z \geq \sigma^2\epsilon-\frac{\Delta^2}{2}\right) \\
    &=  P(  Z \geq \frac{\sigma\epsilon}{\Delta}-\frac{\Delta}{2\sigma}) 
    =  \Phi\left( -\frac{\sigma\epsilon}{\Delta}+\frac{\Delta}{2\sigma}\right) 
\end{align*}
This is the upper bound we needed of the privacy loss random variable. We similar need a lower bound on the privacy loss random variable. The steps are very similar, we have that,
\begin{align*}
    P(PL\leq \epsilon) &= P\left(\log\left[\frac{\exp\{-\frac{1}{2\sigma}(\|h(D)\|_{\mcH}^2-2T_D(x))\}}{\exp\{-\frac{1}{2\sigma}(\|h(D')\|_{\mcH}^2-2T_{D'}(x))\}}\right]\leq \epsilon\right)\\
     &= P\left(  Z \leq \frac{\sigma\epsilon}{\Delta}-\frac{\Delta}{2\sigma}\right) 
     = \Phi\left(\frac{\sigma\epsilon}{\Delta}-\frac{\Delta}{2\sigma}\right).
\end{align*}
This completes the proof.

\subsection{Additional Results}\label{ss:AddRes}
In \ref{ss:ExRes} we describe how to privatize a point-wise average face. Figure \ref{fig:moreexamples} displays a comparison of two private faces with $\mu_T=2$ and $\mu_T=3$, each from two angles. The two faces look fairly similar, but this is an artifact of the noise making the face rough and fuzzy. The face on the right column has a smoother outline as compared to the face on the left. In Table \ref{tab:Error} we see that the MSE of the face on the right is nearly half of that of the face on the left.

\begin{figure}[t]
    \centering
    \begin{tabular}{|@{}c@{} |@{}c@{}@{}c@{}|}
    \hline
        
        {\includegraphics[trim = 100 0 60 0, clip, height =1.3in]{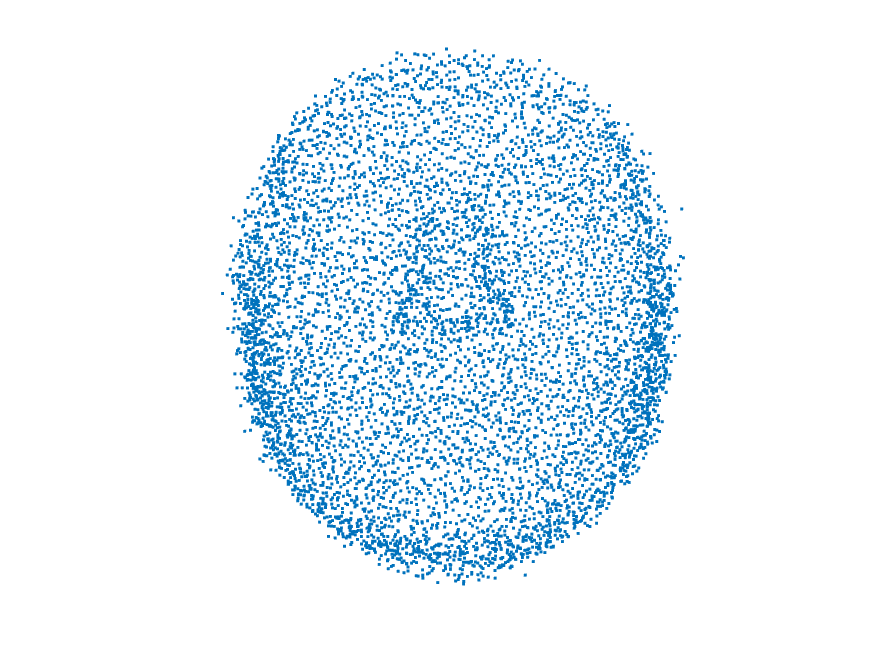}} &
        {\includegraphics[trim = 100 0 60 0, clip, height =1.3in]{PWmeanGDPmu3Ang1}} &
        \\

        {\includegraphics[trim = 100 50 100 50, clip, height =1.3in]{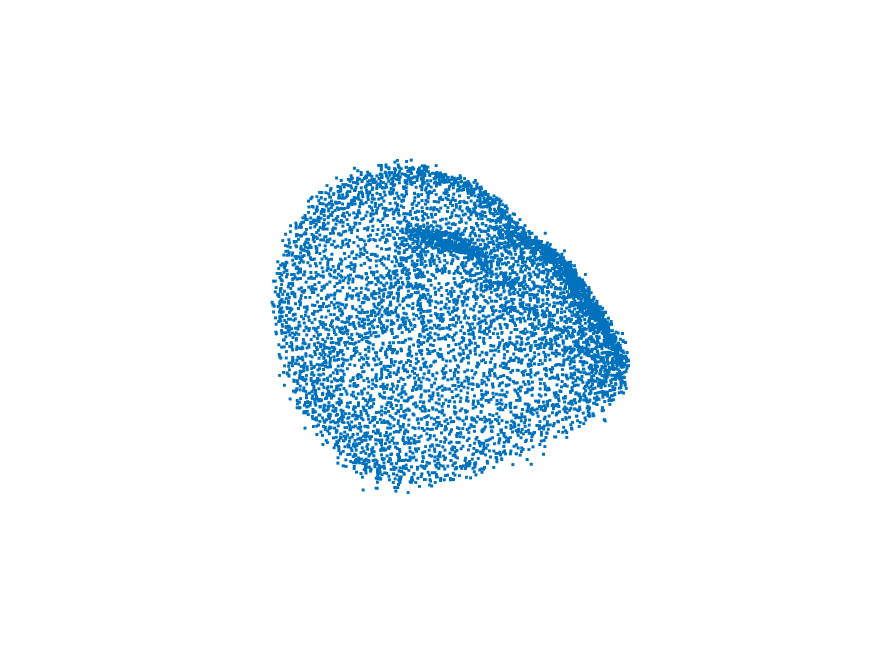}} & 
        {\includegraphics[trim = 100 50 100 50, clip, height =1.3in]{PWmeanGDPmu3Ang2}} &

        \\\hline
    \end{tabular}
    
    \caption{Left: A private face sanitized with Gaussian noise point-wise with total $\mu_T=2$. Right: A private face sanitized with Gaussian noise point-wise with total $\mu_T=3$.}
    \label{fig:moreexamples}
\end{figure}

In the right panel of Figure \ref{fig:facedetail} we show the coordinate curves of one face radial curve for one face; In Figure \ref{fig:facedetails2} we show the coordinate curves of one face radial curve for one hundred faces. We see in Figure \ref{fig:facedetails2} that the curves (and hence the features) are aligned across individuals; this empirically shows the alignment and registration achieved by the methodology in Section \ref{ss:esa}.

\begin{figure}
    \centering
    \begin{tabular}{@{}c@{} @{}c@{} @{}c@{} }
        \fbox{\includegraphics[height=1in]{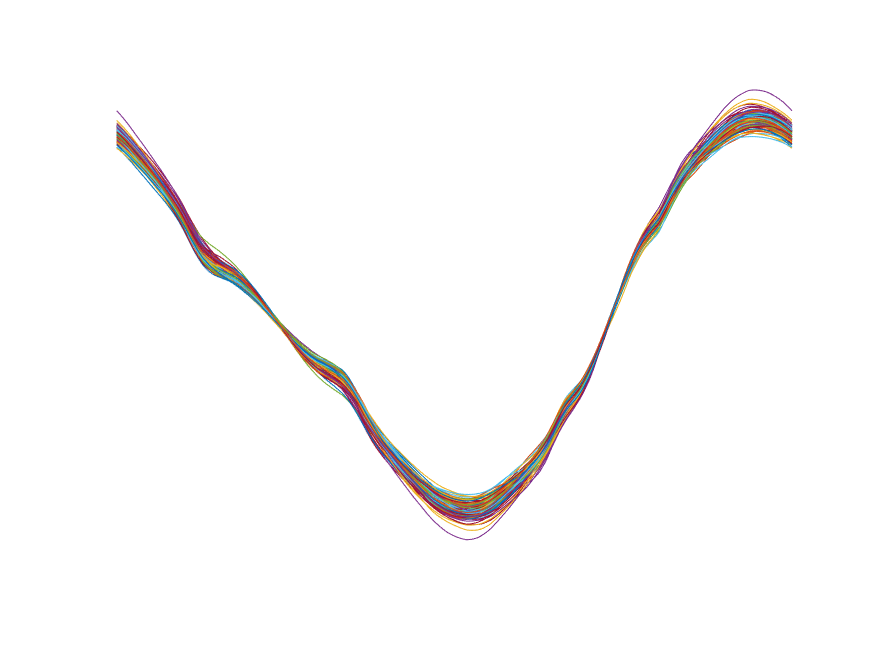} }&
        \fbox{\includegraphics[height=1in]{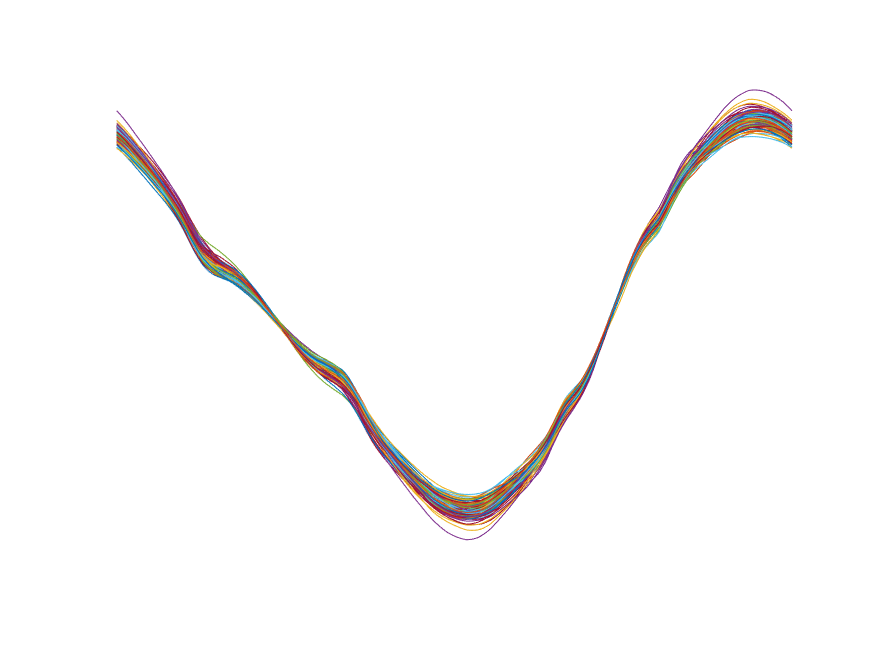} }&
        \fbox{\includegraphics[height=1in]{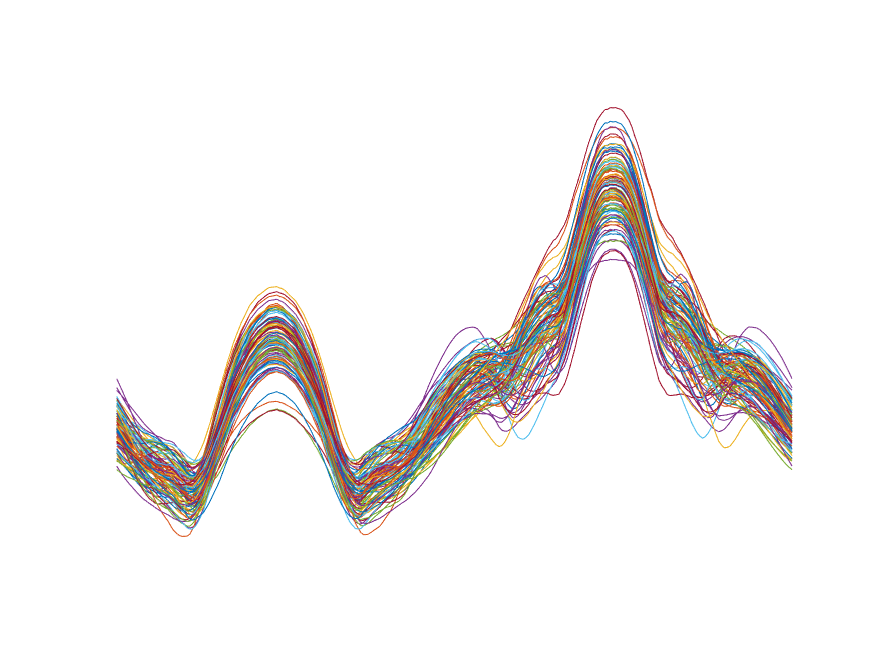} }
    \end{tabular}
        
    \caption{The $x$, $y$, and $z$ coordinate curves, from left to right, of one face radial curve for one hundred individuals.}
    \label{fig:facedetails2}
\end{figure}
\begin{figure}
    \centering
    \includegraphics[width=0.7\linewidth]{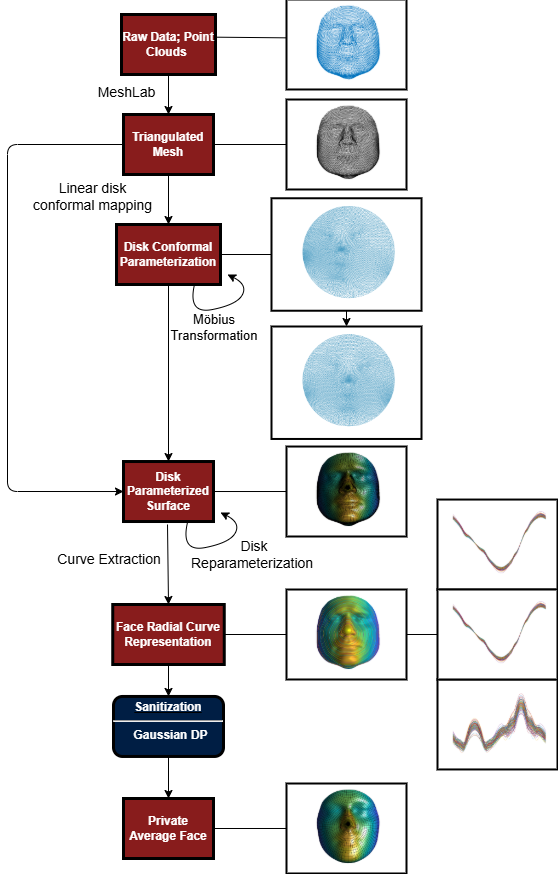}
    \caption{A pipeline diagram of our methodology. We generate disk parameterized surfaces which we register and align to a template face. We extract the face radial curves from the disk parametrization and produce a sanitized mean via our GDP functional data mechanism.}
    \label{fig:diagram}
\end{figure}

% \begin{figure}
%     \centering
%     \begin{tabular}{|@{}c@{}| @{}c@{} |@{}c@{} |@{}c@{} @{}c@{} @{}c@{} }
%     \hline
%     {\includegraphics[trim = 80 0 60 0, clip, height=1in]{pointwisemeanAng1.png}} &

%     \includegraphics[height=1in]{TriangulationV20010.png} &
     
%      \fbox{\includegraphics[trim = 40 0 35 0, clip, height=1in]{ConfMap0010.png}} &

%      \fbox{\includegraphics[trim = 80 20 70 10, clip, height=1in]{FaceRadCurve80.png}} 
%      \\

%      \includegraphics[trim = 80 0 60 0, clip, height = 1in]{pointwisemeanAng2.png} &

%      \includegraphics[height = 1in]{images/TriangulationV10010.png} &

%      \fbox{\includegraphics[trim = 40 0 35 0, clip, height = 1in]{ConfMapMobius0010.png} }\\ 
%      \hline
% \end{tabular}
%     \caption{First Column: }
%     \label{fig:}
% \end{figure}

\end{document}